\documentclass[a4paper, twocolumn, accepted=2024-04-22]{quantumarticle}

\pdfoutput=1

\usepackage[numbers, sort&compress]{natbib}

\usepackage{graphicx}
\usepackage{amsmath}
\usepackage{amssymb}
\usepackage{amsthm}
\usepackage{bbold}
\usepackage{pifont}
\usepackage{array}
\usepackage{tabularx}
\usepackage{hyperref}
\hypersetup{
 pdftitle = {Energy conservation and fluctuation theorem are incompatible for quantum work},
 pdfauthor = {Karen Hovhannisyan and Alberto Imparato},
 breaklinks=true,
 pdfnewwindow=true,
 colorlinks=true,
 linkcolor=red,
 citecolor=blue,
 filecolor=blue,
 urlcolor=blue
}
\usepackage{color}
\usepackage{psfrag}
\usepackage{soul}
\usepackage{textcomp}
\usepackage{bbm}
\usepackage{tikz}

\usetikzlibrary{arrows, shapes.misc, shadows}
\newcommand{\boxA}[1]{\tikz[baseline=(char.base)]\node[anchor = north west, draw, rectangle, rounded corners, inner xsep = 1.5mm, inner ysep = 1.25mm, text height = 2.5mm, color = red, fill = white](char){\ensuremath{#1}} ;}
\newcommand{\boxAa}[1]{\tikz[baseline=(char.base)]\node[anchor = north west, draw, rectangle, rounded corners, inner xsep = 1.1mm, inner ysep = 0.93mm, text height = 2.5mm, color = red, fill = white](char){\ensuremath{#1}} ;}

\newtheorem{theorem}{Theorem}

\newtheorem{lemma}{Lemma}
\newtheorem{constraint}{Constraint}
\newtheorem{proposition}{Proposition}

\DeclareMathOperator{\tr}{Tr}
\DeclareMathOperator{\diag}{diag}

\renewcommand{\appendixname}{APPENDIX}

\newcommand{\comments}[1]{}
\newcommand{\bea}{\begin{eqnarray}}
\newcommand{\eea}{\end{eqnarray}}
\newcommand{\besa}{\begin{subequations}\begin{eqnarray}}
\newcommand{\eesa}{\end{eqnarray}\end{subequations}}
\newcommand{\beaa}{\begin{eqnarray}\begin{aligned}}
\newcommand{\eeaa}{\end{aligned}\end{eqnarray}}

\newcommand{\ket}[1]{| #1 \rangle}

\newcommand{\bra}[1]{\langle #1 |}

\newcommand{\av}[1]{\langle #1 \rangle}

\newcommand{\id}{\mathbbm{1}}
\newcommand{\nul}{\mathbb{0}}
\newcommand{\RN}[1]{\textup{\uppercase\expandafter{\romannumeral#1}}}
\newcommand{\rmT}{\mathrm{T}}

\begin{document}

\title{Energy conservation and fluctuation theorem are incompatible for quantum work}

\author{Karen V. Hovhannisyan}
\email{karen.hovhannisyan@uni-potsdam.de}
\affiliation{University of Potsdam, Institute of Physics and Astronomy, Karl-Liebknecht-Str. 24-25, 14476 Potsdam, Germany}

\author{Alberto Imparato}
\email{imparato@phys.au.dk}
\affiliation{Department of Physics and Astronomy, Aarhus University, Ny Munkegade 120, 8000 Aarhus, Denmark}

\begin{abstract}

Characterizing fluctuations of work in coherent quantum systems is notoriously problematic. Here we reveal the ultimate source of the problem by proving that ($\mathfrak{A}$) energy conservation and ($\mathfrak{B}$) the Jarzynski fluctuation theorem cannot be observed at the same time. Condition $\mathfrak{A}$ stipulates that, for any initial state of the system, the measured average work must be equal to the difference of initial and final average energies, and that untouched systems must exchange \textit{deterministically} zero work. Condition $\mathfrak{B}$ is only for thermal initial states and encapsulates the second law of thermodynamics. We prove that $\mathfrak{A}$ and $\mathfrak{B}$ are incompatible for work measurement schemes that are differentiable functions of the state and satisfy two mild structural constraints. This covers all existing schemes and leaves the theoretical possibility of jointly observing $\mathfrak{A}$ and $\mathfrak{B}$ open only for a narrow class of exotic schemes. For the special but important case of state-independent schemes, the situation is much more rigid: we prove that, essentially, only the two-point measurement scheme is compatible with $\mathfrak{B}$.

\end{abstract}

\maketitle

\section{Introduction}
\label{sec:intro}

Work is a central notion in both mechanics and thermodynamics, being essentially the only touching point between the two theories. When information about the full statistics of work is required, which is especially relevant for small systems, classical mechanics answers by means of trajectories \cite{Bochkov_1977, Jarzynski_1997}, leaving no prescription for accessing the statistics of work in the quantum regime \cite{Bochkov_1977, Kurchan_2000, Tasaki_2000, Allahverdyan_2005, Esposito_2009, Allahverdyan_2014, Talkner_2016, Perarnau-Llobet_PRL_2017, Sampaio_2018, Baumer_2018, Brodier_2020}. This uncertainty has bred a number of approaches towards measuring work in the quantum regime \cite{Bochkov_1977, Kurchan_2000, Tasaki_2000, Yukawa_2000, Allahverdyan_2005, Talkner_2007, Allahverdyan_2014, Solinas_2015, Deffner_2016, Aberg_2018, Alhambra_2016, Talkner_2016, Perarnau-Llobet_PRL_2017, Miller_2017, Xu_2018, Sampaio_2018, Gherardini_2021, Brodier_2020, Beyer_2020, Micadei_2021, Kerremans_2022, Janovitch_2022, Beyer_2022, Pei_2023}. Nevertheless, as amply discussed in the literature \cite{Allahverdyan_2005, Allahverdyan_2014, Talkner_2016, Perarnau-Llobet_PRL_2017, Lostaglio_2018, Baumer_2018}, and in Sec.~\ref{sec:statedep}, they all become problematic in certain regimes.

In this paper, we show that the fundamental source of all these problems is that energy conservation and Jarzynski fluctuation theorem (or simply Jarzynski equality) \cite{Jarzynski_1997} for work cannot be observed at the same time in quantum mechanics. Namely, we prove that essentially no quantum measurement slated to measure work can produce a random variable that can simultaneously satisfy \hypertarget{req:a}{$\boxA{\mathfrak{A}}$} energy conservation for all states and \hypertarget{req:b}{$\boxA{\hspace{-0.2mm}\mathfrak{B}\hspace{-0.3mm}}$} Jarzynski equality (JE) for all thermal states.

Since the JE necessarily holds in classical mechanics, requiring it also in the quantum regime is a way of imposing a weak form of quantum--classical correspondence principle on the distribution of work, without going into the specifics of the quantum--classical transition itself. Thus, elegantly incorporating both the second law of thermodynamics and the correspondence principle, the JE can serve as a high-level filter through which physically meaningful definitions of work fluctuations ought to be able to pass, especially given the foundational role fluctuation theorems play in statistical mechanics \cite{Seifert_2012, Esposito_2009}.

Undeniably, energy conservation is another such filter, and that virtually no definition of work fluctuations can pass through both, means that quantum work cannot be thought of as a classical random variable. In our definition, a measurement scheme obeys energy conservation if the work statistics measured by the scheme reflect energy conservation for the \textit{unmeasured} system. Specifically, the average work performed on an unmeasured thermally isolated system is equal to the difference between the final and initial average energies of the system, and a proper measurement of work must output a random variable the first moment of which coincides with
that value. We call this condition \hypertarget{req:aa}{$\boxAa{\mathfrak{A}_1\hspace{-0.7mm}}$}. Another simple consequence of energy conservation is that, if a system is untouched, then zero work with probability $1$ is performed on it. Therefore, a correct measurement of work must output exactly that statistics; this is condition \hypertarget{req:ab}{$\boxAa{\mathfrak{A}_2\hspace{-0.7mm}}$}. One may call \hyperlink{req:aa}{$\mathfrak{A}_1$}\newcommand{\reqAa}{\hyperlink{req:aa}{$\mathfrak{A}_1$}} the ``average energy conservation'' condition, and to contrast with that, we will call \hyperlink{req:a}{$\mathfrak{A}$}\newcommand{\reqA}{\hyperlink{req:a}{$\mathfrak{A}$}}, which is {\reqAa} and \hyperlink{req:ab}{$\mathfrak{A}_2$}\newcommand{\reqAb}{\hyperlink{req:ab}{$\mathfrak{A}_2$}} imposed jointly, the ``\textit{detailed} energy conservation'' condition.

\medskip

Our analysis is divided into two parts. First, we focus on work measuring schemes that \textit{do not} depend on the initial state of the system. We prove that, for such schemes, \hyperlink{req:b}{$\mathfrak{B}$}\newcommand{\reqB}{\hyperlink{req:b}{$\mathfrak{B}$}} necessitates a conflict with {\reqAa}. Moreover, the schemes that satisfy {\reqB} and yield the correct average for initial thermal states produce the same statistics as the two-point measurement (TPM) scheme \cite{Kurchan_2000, Tasaki_2000, Esposito_2009, Talkner_2007, Campisi_2011}; see Theorem~\ref{thm:JE=TPM}. This establishes that JE~$\Leftrightarrow$~TPM for state-independent schemes.

Then, we ask whether {\reqA} and {\reqB} can be reconciled by using schemes that are allowed to depend on the system's initial state. Considering physically meaningful only those schemes that depend on the state differentiably, we prove that, while {\reqAa} and {\reqB} can be made compatible, when requiring energy conservation to be satisfied to a fuller extent---namely, imposing {\reqAb} alongside {\reqAa}---the compatibility with {\reqB} breaks down for essentially all reasonable state-dependent schemes, leaving the theoretical possibility open for only a narrow exotic class of schemes (Theorem~\ref{thm:no-go}). Interestingly, merely requiring {\reqAa} and {\reqAb} to hold at the same time already severely restricts the class of acceptable schemes (Lemmas~\ref{thm:AaAbClashIndep} and~\ref{thm:AaAbClash}).

\section{Formal setup}
\label{sec:setting}

Throughout this paper, we will focus on closed systems (also called ``conservative'' in classical mechanics). This setting is the most general: any open system is a part of a larger closed system consisting of the system itself and all the external systems with which it is correlated, interacts, and will interact during the process under study \cite{ll5, Lindblad_book}. Arguably, this picture also incorporates the process of quantum measurement \cite{Allahverdyan_2013, Masanes_2019, Abdelkhalek_2016}. For closed systems, a process is described by a time-dependent Hamiltonian $H(t)$ that at the beginning ($t = t_{\mathrm{in}}$) has some value $H$ and at the end ($t = t_{\mathrm{fin}}$) some value $H'$, generating the unitary evolution operator $U = \overrightarrow{\exp}\big[- i \int_{t_{\mathrm{in}}}^{t_{\mathrm{fin}}}dt H(t)\big]$.

For any quantum process, the observed statistics of any quantity (e.g., work) can be described by a positive operator-valued measure (POVM) \cite{mikeike}. Since the whole physical arrangement through which the process is executed and measured is encapsulated in the POVM, by a ``work measuring scheme'' we will understand a POVM and a set of outcomes prescribed to it: $\{ M_W, W \}_W$, where $M_W$ are the POVM elements and $W$ are the associated outcomes. Thus, any random variable representing work is nothing but a (generalized) quantum measurement with outcomes $W$ occurring with probabilities
\begin{align} \label{eq:POVM}
p_W = \tr(\rho M_W),
\end{align}
where $\rho$ is the initial state of the system.

When the measuring device, through which the statistics is observed, does not depend on the system's initial state, then, expectedly, the corresponding POVM is also state-independent \cite{Perarnau-Llobet_PRL_2017}. Surely, the scheme shall depend on the process (i.e., $H$, $H'$, and $U$). In general, however, both the POVM elements and the outcomes will also depend on the system's initial state. Such state-dependent schemes comprise the most general set of work measurements.

\medskip

In these terms, {\reqAa} means that
\beaa \label{eq:aver}
\sum_W W \tr(\rho M_W) &= \tr(U \rho \, U^\dagger H') - \tr(\rho H)
\\
&= \tr(\rho \, \Omega)
\eeaa
is required to hold for all $\rho$'s. Here, $\tr(\rho H)$ and $\tr(U \rho U^\dagger H')$ are, respectively, the initial and final average energies, and
\begin{align} \label{hatW:def}
\Omega = U^\dagger H' U - H
\end{align}
is the so-called Heisenberg operator of work (HOW) \cite{Allahverdyan_2005}.

Condition {\reqB} imposes
\begin{align} \label{eq:jarz}
\sum_W e^{-\beta W} \tr(\tau_\beta M_W) = e^{-\beta (F_\beta [H'] - F_\beta [H])},
\end{align}
for any $\beta > 0$, where $\tau_\beta = \frac{1}{Z_\beta [H]} e^{-\beta H}$ is the thermal state, with $Z_\beta[H]=\tr e^{-\beta H} $ and $F_\beta[H]= -\beta^{-1}\ln Z_\beta [H]$ being, respectively, the partition function and equilibrium free energy.

Lastly, {\reqAb} concerns systems that are untouched, i.e., when $H(t) = H = \mathrm{const}$. In such a case, {\reqAb} demands that $p_{W=0} = \tr(\rho M_{W=0}) = 1$ for any density matrix $\rho$.

\subsection{Discussion of the setup}
\label{sec:setup_discuss}

\begin{table}[t!]
\caption{Relationship of all existing schemes with conditions \textcolor{red}{$\mathfrak{A}_1$}, \textcolor{red}{$\mathfrak{A}_2$}, and \textcolor{red}{$\mathfrak{B}$}. Schemes satisfying only one of these conditions \cite{Deffner_2016} or featuring quasiprobabilities \cite{Allahverdyan_2014, Solinas_2015, Miller_2017, Xu_2018, Brodier_2020, Pei_2023} are not shown. Also not shown are Refs.~\cite{Aberg_2018, Alhambra_2016}, since their framework does not produce any work statistics for states with coherences in the energy eigenbasis. We see that none of the schemes simultaneously satisfies all three conditions.}
\label{tab:literature}
\begin{tabular}{ l c c c }
  \hline \hline
  Scheme & \, {\reqAa} & \, {\reqAb} & \, {\reqB}
\\
  \hline
  \cite{Bochkov_1977, Allahverdyan_2005}, \, \cite{Sampaio_2018}, \, \cite{Beyer_2020}, \, \cite{Kerremans_2022} & \, \ding{52} & \, \ding{52} & \, \ding{56}
\\
  \cite{Gherardini_2021}, \, \cite{Micadei_2021}, \, \cite{Beyer_2022}, \, Appendix~\ref{app:scheme} & \, \ding{52} & \, \ding{56} & \, \ding{52}
\\
  \cite{Kurchan_2000, Tasaki_2000}, \, \cite{Yukawa_2000}, \, $\Upsilon_\rho$ [Eq.~\eqref{rotow}] & \, \ding{56} & \, \ding{52} & \, \ding{52} \\
  \hline \hline
\end{tabular}
\end{table}

To provide additional context, we show in Table~\ref{tab:literature} the relationship of all existing schemes with Conditions {\reqAa}, {\reqAb}, and {\reqB}; we see that none of them satisfies all three at the same time. Additionally, for the sake of the forthcoming analysis, below we briefly review the two paradigmatic work measurement schemes, TPM \cite{Kurchan_2000, Tasaki_2000} and HOW \cite{Bochkov_1977, Allahverdyan_2005}, in the light of these conditions.

In the TPM scheme \cite{Kurchan_2000, Tasaki_2000}, one first measures the energy at the beginning, obtaining the outcome $E_a$ and post-measurement state $P_a \rho P_a / \tr(\rho P_a)$ with probability $\tr(\rho P_a)$, where $E_a$ and $P_a$ are, respectively, the eigenvalues and eigenprojectors of $H$ ($H = \sum_a E_a P_a$). Then, the unitary process is implemented and the energy of the system is measured again. This yields the outcome $E'_k$ with conditional probability $\tr\left[U \frac{P_a \rho P_a}{\tr(\rho P_a)} U^\dagger P'_k \right]$, where $E'_k$ and $P'_k$ are the eigenvalues and eigenprojectors of $H'$. Thus, according to the scheme, the outcomes of work and their probabilities are
\begin{align} \label{tpem}
W_{ak} = E'_k - E_a, \qquad p_{a k}^{\mathrm{TPM}} = \tr\big(\rho M_{ak}^{\mathrm{TPM}}\big),
\end{align}
where
\begin{align} \label{salame}
M_{a k}^{\mathrm{TPM}} = P_a U^\dagger P'_k U P_a.
\end{align}
Noting that $\big[M_{a k}^{\mathrm{TPM}}, H\big] = \nul$ for all $a$ and $k$, we see that $\sum W_{a k} M_{a k}^{\mathrm{TPM}}$ commutes with $H$ for any $H$, $H'$, and $U$; here $\nul$ is the zero operator. Whereas there exist such $H$, $H'$, and $U$ for which $[U^\dagger H' U, H] \neq \nul$ and hence $\Omega$ does not commute with $H$, and therefore $\Omega$ cannot be equal to $\sum W_{a k} M_{a k}^{\mathrm{TPM}}$. In such a case, there will exist a $\rho$ for which Eq.~\eqref{eq:aver} is violated, which shows that the TPM scheme is incompatible with {\reqAa} \cite{Allahverdyan_2005, Talkner_2016, Perarnau-Llobet_PRL_2017, Baumer_2018}.
Nonetheless, we note that the TPM scheme always satisfies {\reqAb}. Moreover, in the classical limit, where all the commutators go to zero, the discrepancy with {\reqAa} also vanishes (see Refs.~\cite{Jarzynski_2015, Garcia-Mata_2017, Funo_2018, Fei_2018} for more details about the TPM scheme's classical limit).

In the HOW scheme \cite{Allahverdyan_2005}, the statistics of work is identified with the measurement statistics of the operator $\Omega$. Therefore, by construction, the scheme satisfies both requirements {\reqAa} and {\reqAb}. However, {\reqB} is not satisfied because, in general,
\begin{align} \nonumber
\big\langle e^{-\beta W} \big\rangle_{\mathrm{HOW}} := \tr\big(\tau_\beta \, e^{-\beta \Omega} \big) \geq e^{-\beta (F_\beta [H'] - F_\beta [H])},
\end{align}
which straightforwardly follows from the Golden--Thompson inequality \cite{Petz_1994}. When the system's Hilbert space is finite-dimensional, the inequality is strict when (and only when) $\beta [\Omega, H] = \beta [U^\dagger H' U, H] \neq \nul$ \cite{Petz_1994}. However, the HOW scheme satisfies {\reqB} in the classical limit. In finite-dimensional spaces, this is a simple consequence of all the commutators vanishing in the classical limit \cite{Allahverdyan_2005}. In continuous-variable systems, issues with continuity may arise, so it is not mathematically guaranteed anymore. Nonetheless, for a harmonic oscillator, by performing an explicit calculation, we show in Supplementary Note~\ref{suppl:HOW_clasl} that the HOW scheme indeed satisfies {\reqB} in the $\hbar \to 0$ limit.



That {\reqAa}, {\reqAb}, and {\reqB} become compatible in the classical limit for the TPM and HOW schemes (and in fact for some other schemes too \cite{Pan_2019, Brodier_2020}), is yet another indication that the classical approach to determining the statistics of work is unlikely to provide useful guidance for the quantum regime.

\section{Results}

\subsection{State-independent schemes}
\label{sec:JEclass}

Let us define the JE class to be the set of all state-independent schemes that satisfy {\reqB}, i.e., the Jarzynski equality for all thermal states. The following lemma, proven in Appendix~\ref{app:Lemma1proof}, gives a
complete characterization of the JE class.

\begin{lemma}[Characterization of the JE class] \label{thm:JEclass}
A state-independent scheme $\{ (M_W, W) \}_W$ with a discrete set of outcomes $\{W\}$ satisfies {\reqB} \textit{iff}
\begin{align} \label{JC3}
\{ W \}_W = \{E'_k - E_a\}_{a = 1}^A\phantom{|}_{k = 1}^K,
\end{align}
so that the outcomes and the corresponding POVM elements can be labeled as $W_{a k} = E'_k - E_a$ and $M_{a k}$, along with
\begin{align} \label{bresaola}
\sum_{\substack{a, \, b, \, m; \\ E'_m - E_b + E_a = E'_k}} \tr(M_{b m} P_a) = \tr P'_k
\end{align}
and
\begin{align} \label{mortadella}
\tr(M_{b m} P_a) = 0 \;\;\, \mathrm{if} \;\;\, E'_m - E_b + E_a \not\in \{E'_k\}_{k=1}^K.
\end{align}
\end{lemma}

A key consequence of Lemma~\ref{thm:JEclass}, also proven in Appendix~\ref{app:Lemma1proof}, is the following no-go theorem.

\begin{theorem}[No-go for state-independent schemes] \label{thm:indep_no-go}
If a state-independent scheme satisfies {\reqB}, then it is incompatible with {\reqAa}.
\end{theorem}

Note that the JE class strictly contains the TPM scheme, which is expressed in the fact that Eqs.~\eqref{bresaola} and~\eqref{mortadella} [and even the more specific Eqs.~\eqref{JC4} and~\eqref{JC5}] allow for more general POVMs. However, if in addition to {\reqB}, we merely require the scheme to respect {\reqAa} \textit{on thermal states,} the following theorem, proven in Appendix~\ref{app:proofJE=TPM}, will hold.

\begin{theorem}[JE $\Leftrightarrow$ TPM] \label{thm:JE=TPM}
For the (generic) case of nondegenerate $H$ and nondegenerate set of outcomes, the schemes in the JE class that satisfy Eq.~\eqref{eq:aver} for thermal initial states are equivalent to the TPM scheme.
\end{theorem}

We emphasize that here we do not require Eq.~\eqref{eq:aver} to hold for nonthermal states. Theorem~\ref{thm:JE=TPM} shows that the TPM scheme is essentially uniquely determined by that, for thermal states, it reproduces the correct average work and Jarzynski equality. This relation can be expressed as JE~$\Rightarrow$~TPM. The reverse, namely, that the TPM scheme satisfies the JE (i.e., JE~$\Leftarrow$~TPM), has been known from the very conception of the scheme \cite{Kurchan_2000, Tasaki_2000}, and was in fact what motivated its experimental implementations \cite{Huber_2008, Batalhao_2014, An_2015}. Theorem~\ref{thm:JE=TPM} thus means that, basically, JE~$\Leftrightarrow$~TPM for state-independent schemes.

Note that this equivalence establishes a crucial connection between the main result of Ref.~\cite{Perarnau-Llobet_PRL_2017} and the fluctuation theorem, thereby significantly extending and generalizing the physical scope of the former. In Ref.~\cite{Perarnau-Llobet_PRL_2017}, it was proven that, if a scheme produces the same statistics as the TPM scheme on diagonal states, then it simply coincides with the TPM scheme. With Theorem~\ref{thm:JE=TPM}, we now know that the a priori less restrictive and more physically meaningful condition of satisfying energy conservation and JE on thermal states is already sufficient to ensure identicality with the TPM scheme.

\subsection{State-dependent schemes}
\label{sec:statedep}

Having excluded the possibility of simultaneously satisfying the conditions {\reqAa} and {\reqB} with work measuring schemes that are limited to being state-independent, let us explore the opportunities when no such limitation is posed. That is, when the work-measuring POVM, $\{ M_W \}$, and the corresponding set of outcomes, $\{W\}$, are allowed to depend not only on the process (i.e., $U$, $H$, and $H'$), but also on the system's initial state, $\rho$. These schemes comprise the largest set of measurement protocols allowed by quantum mechanics.

On the one hand, this dramatically increases the set of POVMs to choose from. On the other hand, this happens at the cost of abandoning ``universality''---in order to implement such a scheme, one has to acquire additional knowledge about incoming states and readjust the measurement setup accordingly, which is inconvenient from the practical standpoint.
Nonetheless, state-dependent schemes can be interesting from practical point of view. First of all, it not uncommon in experiments to have some information about how the system is prepared before the work exchange protocol is implemented, which means that one has some knowledge of the initial state. This is particularly relevant when the protocol is designed to extract work---one needs to at least partially know the initial state to be able to extract a nontrivial amount of work \cite{Allahverdyan_2004, Safranek_2023}. As regards the need to adjust the measurement device, it is standard practice in, e.g., adaptive metrology (see Refs.~\cite{Wiseman_1995, Armen_2002, OBrien_2009, Berni_2015} for an illustrative sampler), and should therefore be feasible also in our context. Lastly, a state-independent measurement performed on multiple copies of the system is equivalent to a state-dependent measurement performed on a single system; such multi-copy schemes for work measurement were explored in Refs.~\cite{Perarnau-Llobet_PRL_2017, Wu_2019, Wu_2020}.

\medskip

Expectedly, with this additional freedom, {\reqAa}, {\reqB}, and {\reqAb} become more compatible. In fact, by allowing for completely arbitrary state-dependent schemes, we can resolve the inconsistency problem altogether. For example, by implementing the following protocol: when $[\rho, H] = \nul$, apply the TPM scheme; when $[\rho, H] \neq \nul$, apply the HOW scheme.
Formally, this satisfies {\reqAa}, {\reqB}, and {\reqAb} simultaneously. However, this scheme is highly discontinuous. Indeed, if one adds an infinitesimal amount of coherence to a state that is diagonal in the energy eigenbasis, the statistics will experience a dramatic jump from having $A K$ outcomes $E'_k - E_a$ to having only $d$ outcomes $\Omega_i$, where $d$ is the system's Hilbert space dimension and $\Omega_i$ are the eigenvalues of $\Omega$ [see Eq.~\eqref{hatW:def}]. Importantly, $\{\Omega_i\}$ is not a subset of $\{E'_k - E_a\}$ whenever $[U^\dagger H' U, H] \neq \nul$. In addition to this discontinuity of the outcomes, a radical change of the probability distribution itself will also take place.

We deem such discontinuous situations pathological, especially because any measurement has a finite resolution (e.g., of how diagonal a state is). Moreover, we note that all the characteristics of the system and the process of our interest here depend on the initial state continuously. At the same time, it is natural to expect of the measurement that is able to faithfully track those characteristics to depend on them continuously. Thus, as the composition of continuous functions is also continuous, the said measurement shall depend on the initial state continuously.

We formalize this intuition by requiring the measurement schemes to depend on the state continuously, at least within a certain convex subset of all states. More specifically, let us consider a closed $\epsilon$-ball in the state space around the infinite-temperature ($\beta = 0$) state $\tau_0 = \id/d$: $\mathcal{B}_{\epsilon} := \{ \sigma \; : \,\; \Vert \sigma - \tau_0 \Vert < \epsilon \}$, where $\epsilon > 0$ and $\Vert \cdot \Vert$ is the standard operator norm (in $\ell_2$ space) \cite{hojo}. Let us moreover choose $\epsilon > 0$ such that all $\rho$'s in $\mathcal{B}_\epsilon$ are positive-definite---such an $\epsilon$ always exists as all the eigenvalues of $\tau_0$ are equal to $1/d > 0$, and the eigenvalues of a Hermitian operator depend continuously on the operator \cite{hojo}. Our main continuity assumption is that there exists such an $\epsilon_0 > 0$ for which all the elements $M_W$ of the POVM, and their corresponding outcomes $W$, are continuous functions of $\rho$ in $\mathcal{B}_{\epsilon_0}$. If the scheme is continuous on the set of all density matrices, then of course the choice of $\epsilon_0$ will be arbitrary.

In Appendix~\ref{app:nogoproof}, we prove that requiring the POVM elements and their outcomes to be differentiable with respect to $\rho$, which is a bit stronger than mere continuity, is so restrictive that the following theorem, which is the main result of this paper, holds.

\begin{theorem}[No-go for state-dependent schemes] \label{thm:no-go}
No measurement scheme that is described by a POVM the elements and outcomes of which are at least once differentiable in $\rho$ in some open ball $\mathcal{B}_{\epsilon_0}$ centered at $\tau_0$, are at least twice differentiable with respect to $U$ at $U = \id$, and satisfy two mild structural constraints~\ref{cons:S1} and~\ref{cons:S2} described in Appendix~\ref{app:nogoproof}, can simultaneously satisfy {\reqAa}, {\reqAb}, and {\reqB}.
\end{theorem}

The structural constraints~\ref{cons:S1} and~\ref{cons:S2} essentially boil down to requiring that the only source of complexness of the matrices $M_W$ in the eigenbasis of $H$ can be the noncommutativity of $U$ and $\rho$ with $H$. All work measurement schemes known to the authors satisfy these constraints.

While requiring continuity and differentiability can hardly be considered a limitation in physics, Constraints~\ref{cons:S1} and~\ref{cons:S2}, however mild they might seem (see Appendix~\ref{app:nogoproof}), do leave a narrow perspective open for the existence of a state-continuous scheme that could simultaneously satisfy {\reqAa}, {\reqAb}, and {\reqB}. Of course, there might simply be a more general proof of incompatibility that would not require the Constraints~\ref{cons:S1} and~\ref{cons:S2}. We defer the explicit description of these constraints to Appendix~\ref{app:nogoproof} because there is a certain amount of notation that needs to be established for stating them.

Importantly, while not being able to provide full compatibility of {\reqAa}, {\reqAb}, and {\reqB}, continuously state-dependent schemes do make those conditions more compatible. Indeed, with state-independent schemes, we could simultaneously satisfy {\reqAa} and {\reqAb} (the HOW scheme) as well as {\reqAb} and {\reqB} (the TPM scheme), but not {\reqAa} and {\reqB}. Now, with the additional freedom, in Appendix~\ref{app:scheme}, by combining the ordinary (``forward'') TPM scheme with what we call ``reverse'' TPM scheme, we construct a state-dependent measurement that simultaneously satisfies {\reqAa} and {\reqB} (but not {\reqAb}). Thus, all the requirements are pairwise compatible when continuous state-dependent POVMs are allowed. While this paper was in preparation, three more state-dependent schemes simultaneously satisfying {\reqAa} and {\reqB} (but, again, not {\reqAb}) were reported in Refs.~\cite{Gherardini_2021, Micadei_2021, Beyer_2022}.

Complementing Theorem~\ref{thm:no-go}, the following lemma shows that, strikingly, the two manifestations of the energy conservation, {\reqAa} and {\reqAb}, are incompatible in a wide range of practically relevant situations.

\begin{lemma} \label{thm:AaAbClashIndep}
If, for some state $\rho$ with nonzero coherences ($[\rho, H] \neq \nul$), the outcomes of a work measuring scheme do not depend on the evolution operator $U$, then the scheme cannot simultaneously satisfy conditions {\reqAa} and {\reqAb}, provided the POVM elements of the scheme are differentiable at least once with respect to $U$ at $U = \id$.
\end{lemma}

Lemma~\ref{thm:AaAbClashIndep} follows as a corollary from Lemma~\ref{thm:AaAbClash} in Appendix~\ref{app:nogoproof}. It covers a large class of schemes that includes the TPM scheme, the multi-copy schemes in Refs.~\cite{Perarnau-Llobet_PRL_2017, Wu_2019, Wu_2020}, the scheme we introduced in Appendix~\ref{app:scheme}, and those later proposed in Refs.~\cite{Gherardini_2021, Micadei_2021}. This conflict between the two manifestations of energy conservation, {\reqAa} and {\reqAb}, opens up a completely new perspective on the limited applicability of these schemes for states with coherences in the energy eigenbasis.

Notably, Lemma~\ref{thm:AaAbClashIndep} does not cover the HOW scheme, as all its outcomes depend on $U$; and indeed, it satisfies both {\reqAa} and {\reqAb} by design, but, as discussed in Sec.~\ref{sec:setting}, it does not satisfy {\reqB}. In fact, the HOW scheme is not covered also by the stronger Lemma~\ref{thm:AaAbClash} (see Appendix~\ref{app:nogoproof}).

In a sense, the strategy of our proof of Theorem~\ref{thm:no-go} in Appendix~\ref{app:nogoproof} is inspired by this example, in that we start with requiring {\reqAa} and {\reqAb}, and then show that the structure enforced by them cannot be made compatible with {\reqB}. 

\medskip

To give a non-trivial example to apply Theorem~\ref{thm:no-go} to, let us, inspired by Ref.~\cite{Yukawa_2000}, introduce the state-dependent operator of work
\begin{align} \label{rotow}
\Upsilon_\rho = - \hat{\beta}(\rho)^{-1} \ln \! \big[ e^{\hat{\beta}(\rho) H/2} e^{-\hat{\beta}(\rho) U^\dagger H' U} e^{\hat{\beta}(\rho) H /2} \big], ~
\end{align}
where $\hat{\beta}(\rho) := \arg\min\limits_\beta \Vert \rho - \tau_\beta \Vert$. This defines an ``operator scheme''---the POVM elements are the eigenprojectors of the operator, and the outcomes are the corresponding eigenvalues. This scheme is obviously designed to satisfy {\reqB} and, when $[U^\dagger H' U, H] = \nul$, $\Upsilon_\rho$ coincides with the Heisenberg operator of work $\Omega$ [defined in Eq.~\eqref{hatW:def}], which also means that $\Upsilon_\rho$ has a proper classical limit. One also immediately notes that this scheme satisfies {\reqAb}. In Appendix~\ref{app:rotateW}, we show that this scheme is indeed covered by Theorem~\ref{thm:no-go}, and hence, it cannot satisfy {\reqAa} because it already satisfies {\reqAb} and {\reqB}. And indeed, in Appendix~\ref{app:rotateW}, we explicitly show that $\Upsilon_\rho$ violates {\reqAa} even for thermal states. In addition to that, to mark its debut, we prove some other curious facts about $\Upsilon_\rho$, unrelated to the main point of this paper.


\section{Discussion}

In this paper, we unearthed the fundamental reason why, despite numerous attempts in the literature, no satisfactory measurement scheme for quantum work has been found so far. By taking an abstract approach, we asked whether any quantum measurement, defined as broadly as possible, can produce statistics that are consistent with energy conservation for the unmeasured system and Jarzynski equality.

These conditions ensure that the measurement results reflect the reality of the unmeasured system. Say, nothing happens to the system---there are no systems with which it could exchange energy. Even without measuring the system, energy conservation law already tells us that, with probability $1$, exactly $0$ work is performed on the system. Therefore, any meaningful measurement of work, when applied to the system, should output $0$, with probability $1$, even if the measurement process disrupts the system. This was our condition {\reqAb}.

By the same logic, when we take a system in a known state and drive it unitarily according to a protocol we fully control, then, even without measuring the system, energy conservation tells us that the average work performed on the system is equal to the difference between the final and initial average energies of the system. Thus, we imposed the condition {\reqAa} that a proper measurement of work should output a random variable the first moment of which coincides with that average, no matter how disruptive the measurement may be.

Note that {\reqAa} and {\reqAb} are not conditions about the measurement process itself. Namely, we do not ask whether energy is conserved during the joint evolution of the system, apparatus, and external driving agents.

The main result of this paper, Theorem~\ref{thm:no-go}, is that no reasonable quantum measurement can output such a random variable that would respect the above two manifestations of energy conservation and, at the same time, satisfy the Jarzynski equality whenever the system is in a thermal state (condition {\reqB}). This uncovers a deep connection between energy conservation law and fluctuation theorems that is present \textit{only} in the quantum regime.

Our proof merely assumes (i) differentiability of the measurement operators and outcomes with respect to the state and (ii) that, in the eigenbasis of $H$, complexness in $M_W$ can originate only from noncommutativity of $U$ and $\rho$ with $H$. While (i) can hardly constitute a loophole in any conceivable physical situation, breaking (ii) could in principle provide a leeway for devising a meaningful work measurement scheme. The possibility of finding a more general proof that would not require (ii) is not excluded and is an interesting problem for future research.

Exploring the extent to which state-dependent schemes can be helpful in the related problems of quantum back-action evasion \cite{Wu_2019, Wu_2020}, measuring fluctuations of heat \cite{Levy_2020, Mohammady_2020}, current \cite{Nazarov_2003, Hofer_2017, Hovhannisyan_2019}, entropy production \cite{Elouard_2017, Manzano_2018, Chiara_2022}, and energetic cost of measurements \cite{Sagawa_2009, Abdelkhalek_2016, Elouard_2017, Mohammady_2019, Guryanova_2020, Mohammady_2022}, is an interesting research avenue.

\section*{\hspace{-5mm} Acknowledgments}

We thank Raffaele Salvia for useful feedback on the manuscript. K.V.H. also thanks Armen Allahverdyan, Antonio Ac\'{i}n, and Mart\'{i} Perarnau-Llobet for insightful discussions during early stages of this work. K. V. H. acknowledges support by the University of Potsdam startup funds. A. I. acknowledges support from the Novo Nordisk Foundation NERD grant (Grant no. NNF22OC0075986).

\bigskip

\appendix

\setcounter{equation}{0}
\renewcommand{\theequation}{\thesection\arabic{equation}}
\setcounter{theorem}{0}
\renewcommand{\thetheorem}{\thesection.\arabic{theorem}}
\setcounter{lemma}{0}
\renewcommand{\thelemma}{\thesection.\arabic{lemma}}
\setcounter{proposition}{0}
\renewcommand{\theproposition}{\thesection.\arabic{proposition}}

\section{Proofs of Lemma~\ref{thm:JEclass} and Theorem~\ref{thm:indep_no-go}}
\label{app:Lemma1proof}

To prove Lemma~\ref{thm:JEclass}, let us rewrite Eq.~\eqref{eq:jarz} as
\begin{align} \label{hac}
\sum_W \tr(M_W e^{-\beta H})e^{-\beta W} = \tr e^{-\beta H'}\quad \forall \beta,
\end{align}
where $\{ M_W \}_W$ is a POVM that does not depend on the state of the system, but may depend on the process (i.e., $H$, $H'$, and $U$).

Recalling the eigenresolutions $H = \sum_{a=1}^A E_a P_a$ and $H' = \sum_{k=1}^K E'_k P'_k$ and introducing the degeneracies
\begin{align}
g_a := \tr P_a \qquad \mathrm{and} \qquad g'_k := \tr P'_k,
\end{align}
we obtain from Eq.~\eqref{hac} that
\begin{align} \label{lavash}
\sum_{W, a} \tr(M_W P_a)e^{-\beta (E_a + W)} = \sum_{k=1}^K g'_k e^{-\beta E'_k}
\end{align}
must hold for all $\beta$'s. Whence, using the expansion $e^x = \sum_{N=0}^\infty \frac{x^N}{N!}$, we find that
\begin{align} \label{JC2}
\sum_{W, a} \tr(M_W P_a)(E_a + W)^N = \sum_{k=1}^K g'_k (E'_k)^N 
\end{align}
for all $N \in \mathbb{N}$.

Keeping in mind that everything above is invariant under a global constant energy shift, we choose the ground state of $H'$ to have zero energy. Moreover, since the eigenresolution $H' = \sum_{k=1}^K E'_k P'_k$ already accounts for possible degeneracies, all values of $E'_k$ are distinct. Therefore, the order of the eigenvectors of $H'$ can be chosen such that $0 = E'_1 < E'_2 \cdots < E'_K$. Let us now take $N = 2 L$, $L \in \mathbb{N}$, and divide Eq.~\eqref{JC2} by $(E'_K)^N$. Then, as $L \to \infty$, the right-hand side will converge to $g'_K$. Now, if, for all $a$ and $W$, $|E_a + W| < E'_K$, then the left-hand side will converge to $0$, which cannot be. Likewise, there cannot exist such $a$ and $W$ for which $|E_a + W| > E'_K$, because otherwise the left-hand side would diverge. Therefore, there have to exist at least one pair of $a$ and $W$ such that $|E_a + W| = E'_K$, and $|E_{b} + W'| \leq |E_a + W|$ for all other pairs $b$ and $W'$. And since the equality between the left- and right-hand sides of Eq.~\eqref{JC2} is maintained in the $L \to \infty$ limit, it also holds that
\begin{align} \label{degen}
\sum_{|E_a + W| = E'_K} \tr(M_W P_a) = g'_K.
\end{align}
Moreover, if there exist such $b$ and $W'$ that $E_b + W' = - E'_K$, then necessarily $\tr(M_{W'} P_b) = 0$. Indeed, taking $N = 2 L + 1$, dividing Eq.~\eqref{JC2} by $(E'_K)^N$, and taking the $L \to \infty$ limit, we will obtain the same $g'_K$ on the right-hand side, but $\tr(M_{W'} P_b)$ will enter the left-hand side with a negative sign, which means that equality between the left- and right-hand sides is possible only if $\tr(M_{W'} P_b) = 0$. In other words, for all $N \in \mathbb{N}$, $g'_K (E'_K)^N$ is equal to the sum of all the terms on the left-hand side with maximal $E_a + W$. Thus, for any $N$, we can eliminate those terms from both sides of Eq.~\eqref{JC2} and reiterate the arguments above. Doing so $K$ times, we will conclude that the set of outcomes of work coincides with energy differences. 

This procedure also shows that, if $E_a + W$ does not coincide with any of $E'_k$'s, then $\tr(M_W P_a) = 0$ [check Eq.~\eqref{JC2}]; this proves Eq.~\eqref{mortadella}. Finally, the fact that Eq.~\eqref{degen} is obtained at each step of the iteration (i.e., for all values of $k$) proves Eq.~\eqref{bresaola}, which thereby concludes the proof of Lemma~\ref{thm:JEclass}.

\medskip

Now, let us prove Theorem~\ref{thm:indep_no-go}. For a state-independent scheme to satisfy {\reqAa}, Eq.~\eqref{eq:aver} must be satisfied for any state $\rho$, which means that $\sum_W W M_W = \Omega$ must hold.

On the other hand, Lemma~\ref{thm:JEclass}, and more specifically, Eq.~\eqref{mortadella}, rather strongly restricts the set of eigensubspaces of $H$ among which the matrices $M_{ak}$ are allowed to have coherences. So much so that, for any nontrivial $H$ and $H'$, there exists a $U$ such that $\Omega$ has coherences where $\sum_W W M_W$ cannot.

For brevity of presentation, let us conduct the proof of this fact for the case of nondegenerate set of work outcomes, namely, when
\begin{align} \label{doublenondeg}
E'_k - E_a = E'_m - E_b \;\; \Longrightarrow \;\; k = m \;\; \mathrm{and} \;\; a = b. ~~~
\end{align}
Note that this case is the generic one---those $H$ and $H'$ for which the set of energy differences has degeneracies constitute a $0$-measure subset in the space of all $H$'s and $H'$'s.

Now, for this configuration, Eq.~\eqref{mortadella} simply states that $\tr(M_{a k} P_b) = 0$ when $b \neq a$, which means that, for any $k$, $M_{a k}$ belongs to the $a$'th eigensubspace of $H$; in other words,
\begin{align} \label{JC4}
M_{a k} = P_a M_{a k} P_a.
\end{align}
Hence $[M_{a k}, H] = \nul$ for any $a$ and $k$, and therefore $\sum_W W M_W$ commutes with $H$. Thus, $\sum_W W M_W$ cannot be equal to $\Omega$ whenever $[U^\dagger H' U, H] \neq \nul$.

Lastly, note that, with Eq.~\eqref{JC4} taken into account, Eq.~\eqref{bresaola} simplifies to
\begin{align} \label{JC5}
\sum_a \tr M_{a k} = \tr P'_k. 
\end{align}
A special case of this equation was proven in Ref.~\cite{Ito_2019} for a certain class of two-point energy measurement schemes.

\section{Proof of Theorem~\ref{thm:JE=TPM}}
\label{app:proofJE=TPM}

To prove Theorem~\ref{thm:JE=TPM}, let us, taking into account Eqs.~\eqref{JC3} and~\eqref{JC4}, rewrite Eq.~\eqref{eq:aver} with $\rho = \tau_\beta$ as
\begin{align} \nonumber
\sum_a e^{-\beta E_a} \big[\sum_k E'_k & \tr(P_a U^\dagger P'_k U) - g_a E_a \big]
\\ \nonumber
&= \sum_{a,k} e^{-\beta E_a} (E'_k - E_a) \tr M_{a k},
\end{align}
which has to hold for any $\beta$. Keeping in mind that no two $E_a$'s are the same, and setting $\beta = 0, 1, \dots, A - 1$, we will arrive at a linear system of equations with the coefficient matrix being a Vandermonde matrix \cite{hojo} with a non-zero determinant. Hence, for any $a$,
\begin{align} \label{JCTPM1}
\sum_k (E'_k - E_a) \tr(P_a U^\dagger P'_k U) = \sum_k (E'_k - E_a) \tr M_{a k},
\end{align}
where we have noted that $\sum_k \tr(P_a U^\dagger P'_k U) = g_a$.

At this point, we assume that the operators $M_{a k}$ do not depend on $E'_k$'s within a small neighborhood. This assumption is natural for two reasons. First, the unitary evolution operator, $U = \overrightarrow{\exp}\big[- i \int_{t_{\mathrm{in}}}^{t_{\mathrm{fin}}}dt H(t)\big]$ does not change when values of $H(s)$ are changed on a measure-zero set, and since changing the eigenvalues of $H' = H(t_{\mathrm{fin}})$ is such an operation, $U$ remains intact. Second, since $0 = E'_1 < E'_2 < \cdots < E'_K$ and $K$ is finite, there exists an $\epsilon > 0$ such that the $\epsilon$-neighborhoods of $E'_k$'s are non-overlapping; we will assume $E'_k$-independence in these neighborhoods. Now, before the last-moment change in the eigenvalues of $H'$, the process runs without ``knowing'' about the change, so if $M_{a k}$'s were going to measure probabilities of work outcomes for $E'_k$'s, then a small variation that preserves the order of $E'_k$'s should not impact the measurement probabilities and hence the operators.

With this assumption, we can differentiate Eq.~\eqref{JCTPM1} with respect to $E'_k$'s, with all other parameters fixed, to find that
\begin{align} \label{JCTPM2}
\tr M_{a k} = \tr(P_a U^\dagger P'_k U)
\end{align}
for all $a$'s and $k$'s.

When the initial Hamiltonian has no degeneracies, i.e., $P_a = \ket{a} \bra{a}$ for all $a$'s, then Eq.~\eqref{JC4} implies that $M_{a k} \propto \ket{a} \bra{a}$, and Eq.~\eqref{JCTPM2} shows that
\begin{align}
M_{a k} = \ket{a} \bra{a} U^\dagger P'_k U \ket{a} \bra{a} = M_{a k}^{\mathrm{TPM}},
\end{align}
where the second equality is due to Eq.~\eqref{salame}. Hence, for nondegenerate $H$'s, the JE class consists of the TPM scheme \textit{only}. For degenerate $H$'s, the statistics produced by the schemes in the JE class are guaranteed to coincide with those of the TPM scheme \textit{only} on diagonal initial states of the form $\rho = \sum_a \rho_a P_a$ [which can be easily checked by using Eqs.~\eqref{JC4},~\eqref{JCTPM2}, and~\eqref{tpem}]. Whenever $\rho$ is not proportional to the identity operator in one of the eigensubspaces of $H$, there will exist a POVM that is in the JE class, satisfies Eq.~\eqref{eq:aver} for thermal states, but is nevertheless such that $\tr(\rho M_{a k}) \neq p_{a k}^{\mathrm{TPM}}$.

We emphasize that the $H$'s and $H'$'s for which the JE class is larger than the TPM scheme comprise a $0$-measure subset of the set of all possible $H$'s and $H'$'s. Which in particular means that, for any $H$ and $H'$ yielding a degenerate set of energy difference, there exist infinitesimal perturbations of $H$ and $H'$ that remove the degeneracy. Therefore, if we ask which schemes satisfy the JE and average work condition for thermal initial states for all processes at least in an arbitrarily small neighborhood of a given process, then the TPM scheme will be the only such scheme. In this sense, the JE class is equivalent to the TPM scheme.


\section{Proof of Theorem~\ref{thm:no-go}}
\label{app:nogoproof}

To prove that certain requirements are incompatible, it is sufficient to prove that there exists at least one state and one process for which those requirements entail contradicting results.

We will look for such a state within the very same ball $\mathcal{B}_{\epsilon_0}$ in which the scheme is continuous. The class of processes among which we will look for the desired process will be cyclic Hamiltonian processes ($H' = H$), the evolution operator generated by which has the form $U = e^{-i x h}$, where $x$ is a real number and $h$ is a Hermitian operator. We will call this class of processes $\mathbf{\Pi}$:
\begin{align} \label{processes}
\mathbf{\Pi} = \big\{ (H, H', U) \; : \,\; H' = H, \; U = e^{- i x h} \big\},
\end{align}
where $h$ spans all finite-norm Hermitian operators living in the system's Hilbert space. This class of processes will allow us to conveniently include {\reqAb} in our analysis, because, for any process in $\mathbf{\Pi}$, we can make the system untouched by simply taking the $x \to 0$ limit. Conveniently, the system becomes untouched also in the limit $h \to H$.

Let us now explore the consequences of imposing {\reqAb}. To do so, for each $\rho$ and $h$, we define
\begin{align} \label{setO}
\mathcal{O}_{\rho, h} = \big\{ W \; : \,\; \lim_{x \to 0} W = 0 \big\}
\end{align}
to be the set of outcomes $W$ that are either $0$ or tend to $0$ as $x \to 0$. The rest of the outcomes will tend to non-zero values in the same limit; we denote the set of those outcomes by $\mathcal{N}_{\rho, h}$:
\begin{align} \label{setN}
\mathcal{N}_{\rho, h} = \big\{ W \; : \,\; \lim_{x \to 0} W \neq 0 \big\}.
\end{align}
Then, the requirement {\reqAb} means that, in the $x \to 0$ limit,
\begin{align} \label{probi}
\sum_{W \in \mathcal{N}_{\rho, h}} p_W \to 0 \quad \mathrm{and} \quad \sum_{W \in \mathcal{O}_{\rho, h}} p_W \to 1,
\end{align}
where $p_W = \tr(\rho M_W)$, as in Eq.~\eqref{eq:POVM}. For positive-definite states ($\rho > \nul$), since all their eigenvalues are strictly positive, Eq.~\eqref{probi} entails
\begin{align} \label{povmi}
\sum_{W \in \mathcal{N}_{\rho, h}} M_W \to \nul \quad \mathrm{and} \quad \sum_{W \in \mathcal{O}_{\rho, h}} M_W \to \id.
\end{align}
We emphasize that, in general, like $M_W$'s, the outcomes $W$ depend on $H$, $h$, $x$, and $\rho$.

If $M_W$ is differentiable at least once with respect to $x$ at $x = 0$, then $p_W = \tr(M_W \rho)$ will also be differentiable in $x$ at $x = 0$. Hence, we can write the Taylor series with the remainder in the Peano form:
\begin{align} \label{peano}
p_W(x) = p_W(0) + x p'_W(0) + o(x),
\end{align}
where the small-o is as per the standard asymptotic notation.

Now, if $p_W \to 0$ as $x \to 0$, then $p_W(0) = 0$. Moreover, $p'_W(0)$ will also have to be zero, because otherwise, for sufficiently small $x$'s, $p_W$ would change its sign when $x \to -x$, whereas $p_W$ must always remain non-negative. Therefore, while maintaining nonnegativity,
\begin{align} \label{probabli}
\lim_{x \to 0} p_W = 0 \quad \Longrightarrow \quad p_W = o(x)
\end{align}
With exactly the same reasoning,
\begin{align} \label{povmabli}
\lim_{x \to 0} M_W = \nul \quad \Longrightarrow \quad M_W = o(x).
\end{align}
As discussed above, $\lim_{x \to 0} p_W = 0$ necessitates $\lim_{x \to 0} M_W = \nul$ whenever $\rho$ is positive-definite. Note that $o(x)$ is allowed to be anything from $x^2 \tilde{m}_W$, with some $\tilde{m}_W \geq \nul$, like in the TPM scheme, to something exotic, say, $|x|^{3/2} \tilde{m}_W$.

Let us now see what {\reqAa} entails for processes in $\mathbf{\Pi}$. Recalling Eqs.~\eqref{eq:aver} and~\eqref{hatW:def}, namely that $\av{W} = \tr(\Omega \rho)$, we can apply the Baker--Hausdorff lemma \cite{Baker_1905} to expand $\Omega$ as
\begin{align} \label{HOWsdorff}
\Omega = x \omega + i \frac{x^2}{2} [h, \omega] + O(x^3),
\end{align}
for $x \ll 1$, where we have introduced the operator
\begin{align} \label{what}
\omega := i [h, H],
\end{align}
which leads us to
\begin{align} \label{eq:avWexpand}
\av{W} = x \tr(\rho \, \omega) + i \frac{x^2}{2} \tr (\rho [h, \omega]) + O(x^3).
\end{align}

At this point, we are ready to prove the following lemma that establishes incompatibility between {\reqAa} and {\reqAb}---the two aspects of energy conservation---in special, but practically relevant, circumstances.

\begin{lemma} \label{thm:AaAbClash}
If, for a state-dependent scheme $\{(M_W, W)\}_W$, there exist a state $\rho$ and a process in $\mathbf{\Pi}$, such that $\tr(\rho [h, H]) \neq 0$, for which the outcomes $W$ either do not go to zero as $x \to 0$ (but also do not diverge), or go to zero as $o(x)$, then the scheme cannot simultaneously satisfy {\reqAa} and {\reqAb}.
\end{lemma}

Note that, since $\tr(\rho [h, H]) = \tr(h [H, \rho])$, the state must have nonzero coherences in the energy eigenbasis; namely, $[\rho, H] \neq \nul$. Obviously, for any such $\rho$, there always exist an infinite amount of $h$'s for which the condition $\tr(\rho [h, H]) \neq 0$ in Lemma~\ref{thm:AaAbClash} is satisfied. Also note that the case of $W$ simply being equal to zero is included in the $o(x)$ notation.

It is worth emphasizing that, for Lemma~\ref{thm:AaAbClash} to hold, the scheme does not have to be continuous in $\rho$.

\begin{proof}[Proof of Lemma~\ref{thm:AaAbClash}]
As we see from Eq.~\eqref{eq:avWexpand}, once we enforce {\reqAa}, the condition $\tr(\rho [h, H]) \neq 0$ of the lemma necessitates that
\begin{align} \label{propencil}
\av{W} \propto x.
\end{align}
On the other hand, from Eqs.~\eqref{probi} and~\eqref{probabli} we see that, enforcing {\reqAb} necessitates that, if $\lim_{x \to 0} W \neq 0$, then $p_W = o(x)$, which means that $W p_W = o(x)$. By the condition of the lemma, if $\lim_{x \to 0} W = 0$, then $W = o(x)$, which means that $W p_W = o(x)$ also in this case. Hence, $\av{W} = \sum_W W p_W = o(x)$, which is in clear contradiction with Eq.~\eqref{propencil}. This thus proves that, under the assumptions of the lemma, {\reqAa} and {\reqAb} are incompatible.
\end{proof}

Note that the HOW scheme provides an important example of a scheme that is \textit{not} covered by Lemma~\ref{thm:AaAbClash}. Indeed, from Eq.~\eqref{HOWsdorff} we observe that all the outcomes of the HOW scheme go to zero $\propto x$, which is in straightforward violation of the conditions of Lemma~\ref{thm:AaAbClash}.

We immediately see that Lemma~\ref{thm:AaAbClashIndep} in the main text follows from Lemma~\ref{thm:AaAbClash} as a simple corollary. Indeed, as noted above, for any $\rho$ such that $[\rho, H] \neq \nul$, there always exists a $h$ such that $\tr(\rho [h, H]) \neq 0$. And if the outcomes do not depend on $U$, it means that they are either zero, which falls under the category of $o(x)$, or they are not zero. Hence, all the conditions of Lemma~\ref{thm:AaAbClash} are met.

\vspace{2 mm}
\noindent
\textbf{Proof of Theorem~\ref{thm:no-go}.} Our proof strategy here will be to impose {\reqAa} and {\reqAb} and show that they lead to results that are not compatible with {\reqB}.

By the assumption of the theorem, all $M_W$'s are differentiable at least twice with respect to $x$. Therefore, the argumentation between Eqs.~\eqref{peano} and~\eqref{povmabli} holds, which means that, in view of Eq.~\eqref{povmi}, $M_W = o(x)$ whenever $W \in \mathcal{N}_{\rho, h}$. Due to the existence of the second derivative, we can write the $o(x)$ term as $x^2 u_W + o(x^2)$, with some operator $u_W \geq \nul$.

Let us now add to the set $\mathcal{N}_{\rho, h}$ all those $W \in \mathcal{O}_{\rho, h}$ for which $M_W \to \nul$ as $x \to 0$ (if such $W$'s exist, of course). By the very same logic as above, for such $W$'s, we will also have $M_W = x^2 u_W + o(x^2)$. The addition of such elements will extend $\mathcal{N}_{\rho, h}$ into
\begin{align}
\mathbf{N}_{\rho, h} = \big\{ W \; : \,\; \lim_{x \to 0} M_W = +\nul \big\} \; \supseteq \mathcal{N}_{\rho, h},
\end{align}
and, for this set, we will have
\begin{align} \label{eq:POVMinN}
W \in \mathbf{N}_{\rho, h} \Rightarrow \bigg\{ \! \begin{array}{l} W = U_W + O(x),
\vspace{1mm}
\\
M_W = x^2 u_W + o(x^2). \end{array}
\end{align}
Both $u_W \geq \nul$ and $U_W$ may depend on $\rho$, $H$, and $h$.

For the complementary set
\begin{align}
\mathbf{O}_{\rho, h} = \big\{ W \; : \,\; \lim_{x \to 0} M_W > \nul \big\} \; \subseteq \mathcal{O}_{\rho, h},
\end{align}
the double-differentiability of $M_W$'s and $W$'s yields
\begin{align} \label{eq:POVMinO}
W \in \mathbf{O}_{\rho, h} \Rightarrow \bigg\{ \! \begin{array}{l} W = x V_W + x^2 v_W + o(x^2),
\vspace{1mm}
\\
M_W = m_W + x \mu_W + x^2 \nu_W + o(x^2), \end{array}
\end{align}
where, by definition, $m_W > \nul$ for all $W \in \mathbf{O}_{\rho, h}$, provided $\mathbf{O}_{\rho, h}$ is not empty. All quantities involved ($m_W$, $\mu_W$, $\nu_W$, $V_W$, $v_W$) may depend on $\rho$, $H$, and $h$. Recall that lack of zeroth-order terms in the expansion of $W$ in Eq.~\eqref{eq:POVMinO} is a consequence of {\reqAb}; see the discussion around Eqs.~\eqref{setO}--\eqref{povmabli} in this regard.

It is worth emphasizing at this point that
\begin{align}
\mathbf{N}_{\rho, h} \cap \mathbf{O}_{\rho, h} = \emptyset \quad \mathrm{and} \quad \mathbf{N}_{\rho, h} \cup \mathbf{O}_{\rho, h} = \{W\}.
\end{align}

Due to the fact that $\sum_W M_W = \id$ must hold for any $x$, we find from Eqs.~\eqref{eq:POVMinN} and~\eqref{eq:POVMinO} that
\begin{align} \label{norm01ord}
\sum_{W \in \mathbf{O}_{\rho, h}} m_W = \id, \qquad \sum_{W \in \mathbf{O}_{\rho, h}} \mu_W = \nul,
\\ \label{norm2ord}
\sum_{W \in \mathbf{O}_{\rho, h}} \nu_W + \sum_{W \in \mathbf{N}_{\rho, h}} u_W = \nul.
\end{align}

Next, invoking {\reqAa}, we see that, on the one hand, $\av{W} = \sum_W W \tr(\rho M_W)$, while on the other hand, Eq.~\eqref{eq:avWexpand} must be true for all $x$. By comparing the first-order (in $x$) terms, we find that
\begin{align} \label{av1ord}
\av{Y(\rho, h, H)}_\rho = \av{\omega}_\rho,
\end{align}
where
\begin{align} \label{yillopisto}
Y(\rho, h, H) := \! \sum_{W \in \mathbf{O}_{\rho, h}} \! V_W m_W,
\end{align}
$\omega$ is defined in Eq.~\eqref{what}, and $\av{X}_\rho := \tr(\rho X)$. By comparing the second-order terms in the Taylor expansion, we further obtain
\beaa \label{av2ord}
\sum_{W \in \mathbf{O}_{\rho, h}} \!\! \big[V_W & \av{\mu_W}_\rho + v_W \av{m_W}_\rho \big]
\\ 
+& \sum_{W \in \mathbf{N}_{\rho, h}} \!\! U_W \av{u_W}_\rho = \frac{1}{2} \av{i [h, \omega]}_\rho \, .
\eeaa

In order for our proof of incompatibility to ``work,'' we will need to slightly constrain the set of POVMs. These are high-level assumptions about the structure of allowed POVMs, and are going to be based solely on the fact that all $M_W$'s and $W$'s double differentiable with respect to $x$ at $x=0$.

\begin{constraint}
\label{cons:S1}
The only source of complexity of $m_W$'s in the eigenbasis of $H$ can be the noncommutativity of $h$ with $H$ and $\rho$ with $H$. More precisely, it is the Hermitian operators $\zeta := i [\rho, H]$ and $\omega = i [h, H]$ that control the complexity of $m_W$'s: the matrices $m_W$ are real whenever $\zeta$ and $\omega$ are real.
\end{constraint}

This is a natural assumption roughly saying that, in the classical regime of all the operators commuting with each other, the measurement operators of the scheme should be real in the common eigenbasis. In fact, what we require is somewhat weaker: the condition is needed only on the level of the first-order expansion in $x$. Note that a real $\omega$ means a purely imaginary $h$, which in turn means a real $U$.

Of course, the scheme may be constructed in such a way that the POVM elements depend on some quantity that does not depend on $\rho$ and $h$ and is represented by a complex matrix in the eigenbasis of $H$; it is this kind of situations that Constraint~\ref{cons:S1} aims to exclude.

The second constraint concerns the two positive-semidefinite operators defined as
\begin{align} \nonumber
Y^+ = \! \sum_{\substack{W \in \mathbf{O}_{\tau, h} \\ V_W > 0}} \!\! V_W m_W \quad \mathrm{and} \quad Y^- = \! \sum_{\substack{W \in \mathbf{O}_{\tau, h} \\ V_W \leq 0}} \!\! |V_W| m_W.
\end{align}
By definition~\eqref{yillopisto}, these operators satisfy
\begin{align} \label{buys}
Y^+ - Y^- = Y.
\end{align}

\begin{constraint}
\label{cons:S2}
For all the Gibbs states $\tau_{\beta} \in \mathcal{B}_{\epsilon_0}$, $\tr(\tau_{\beta} \{ Y^+, Y^- \}) \geq 0$ whenever $U$ is a real matrix in the eigenbasis of $H$.
\end{constraint}

Despite the fact that $Y^+\geq \nul$ and $Y^- \geq \nul$ by construction, their anticommutator will not be positive-semidefinite whenever $[Y^+, Y^-] \neq \nul$. Therefore, Constraint~\ref{cons:S2} may not be guaranteed in general. However, it is easy to see that
\begin{align} \label{huys}
\tr(\{ Y^+, Y^- \}) = 2 \tr( Y^+ \, Y^-) \geq 0,
\end{align}
and therefore it is natural to expect that Constraint~\ref{cons:S2} should be generically satisfied whenever the state is close enough to $\id/d$. There are two specific situations---that are fairly general in and of themselves---where Constraint~\ref{cons:S2} is satisfied by construction.

\textit{Situation 1:} $[Y^+, Y^-] = \nul$ when the state is thermal and $U$ is real. In this case, $\{Y^+, Y^-\} \geq \nul$, and thus Constraint~\ref{cons:S2} is satisfied. $Y^+$ and $Y^-$ will commute, for example, when all $m_W$'s commute with each other. That trivially happens, e.g., when $m_W$'s are the eigenprojectors of some observable. Of course, the scenarios where $m_W$'s commute with each other are much more diverse than that.

An alternative scenario where $[Y^+, Y^-] = \nul$ is when all $|V_W|$'s are equal to each other; call that value $v>0$. Indeed, in such a case, $Y^+ = v m$ and $Y^- = v n$, where
\begin{align} \nonumber
m = \! \sum_{\substack{W \in \mathbf{O}_{\tau, h} \\ V_W > 0}} \!\! m_W \quad \mathrm{and} \quad n = \! \sum_{\substack{W \in \mathbf{O}_{\tau, h} \\ V_W \leq 0}} \!\! m_W.
\end{align}
In view of Eq.~\eqref{norm01ord}, $m + n = \id$, and therefore
\begin{align} \nonumber
[Y^+, Y^-] &= v^2 [m, n] = v^2 [m, \id - m] = \nul.
\end{align}
Note that, in this scenario, individual $m_W$'s do not have to commute with each other.

\textit{Situation 2:} $\lim\limits_{\beta \to 0}\tr(\tau_\beta \{Y^+, Y^-\}) \neq 0$; let us call this value $\xi$. We immediately see that $\xi > 0$. Indeed, $\lim_{\beta \to 0} \tau_\beta = \id/d$, and therefore, $\xi > 0$ follows from the observation~\eqref{huys}. Now, by the assumption of the theorem, both $Y^+$ and $Y^-$ are continuous functions of $\rho$ in $\mathcal{B}_{\epsilon_0}$, and hence, $\tr(\tau_\beta \{Y^+, Y^-\})$ is continuous in $\beta$ in a certain interval around $\beta = 0$. Therefore, by the very definition of the notion of continuity, there exist a $\beta_0 > 0$ such that $\tau_{\beta_0} \in \mathcal{B}_{\epsilon_0}$ and $\forall \, \beta \in [0, \beta_0]$, $\big\vert \tr(\tau_\beta \{Y^+, Y^-\}) - \xi \big\vert \leq \xi/2$. Hence, $\tr(\tau_\beta \{Y^+, Y^-\}) \geq \xi/2 > 0$. Thus, at least for this specific $\beta_0 > 0$, Constraint~\ref{cons:S2} is satisfied.

We emphasize once again that Constraints~\ref{cons:S1} and~\ref{cons:S2} are not tied to the specific forms~\eqref{eq:POVMinN} and~\eqref{eq:POVMinO} that are imposed by {\reqAb}---the possibility of Taylor-expanding $M_W$'s and $W$'s is all that is needed for stating Constraints~\ref{cons:S1} and~\ref{cons:S2}.

Now, let us fix the basis to be the eigenbasis of $H$ and explore some consequences of Eq.~\eqref{av1ord}, the differentiability of $M_W$'s and $W$'s with respect to $\rho$ in $\mathcal{B}_{\epsilon_0}$, and Constraint~\ref{cons:S1}.

Consider the family of $\omega$'s represented by real (and necessarily symmetric due to $\omega$'s Hermiticity) matrices in the eigenbasis of $H$:
\begin{align} \label{omfam}
\omega \; \in \; \mathbb{R}^{d \times d}, \quad \omega^T = \omega, \quad \mathcal{D}_H[\omega] = \nul,
\end{align}
where the operator $\mathcal{D}_H[\bullet]$ acts on its argument as $\mathcal{D}_H[\omega] := \sum_a P_a \omega P_a$, where $P_a$'s are the eigenprojectors of $H$. The last equality in~\eqref{omfam}, saying that $\omega$ is zero in the diagonal subspace of $H$, is an immediate consequence of the fact that $\omega = i [h, H]$. Note that, with this choice of $\omega$'s, $U$ is a real matrix, too.

Next, consider the family of $\rho$'s defined as
\begin{align} \label{rhofam}
\rho = \tau_{\beta_0} + z R,
\end{align}
where $\tau_{\beta_0}$ is some Gibbs state in $\mathcal{B}_{\epsilon_0}$ with $\beta_0 > 0$, $z$ is a real adimensional parameter that is $\ll 1$ to ensure that $\rho \in \mathcal{B}_{\epsilon_0}$; henceforth, we will refer to $\tau_{\beta_0}$ as simply $\tau$ to simplify notation. $R$ is an arbitrary real and symmetric matrix with zeros in the diagonal eigenspace of $H$:
\begin{align} \label{Rfam}
R \; \in \; \mathbb{R}^{d \times d}, \quad R^{\mathrm{T}} = R, \quad \mathcal{D}_H[R] = \nul;
\end{align}
we will furthermore limit the norm of $R$ from above by some constant to ensure that $\rho$ does not fall out of $\mathcal{B}_{\epsilon_0}$.
For this family of $\rho$'s, $\zeta$'s are purely imaginary:
\begin{align}
\zeta = i z Q; \quad Q = [R, H] \; \in \; \mathbb{R}^{d \times d}, \quad Q^{\mathrm{T}} = -Q.
\end{align}
Due to the third condition in Eq.~\eqref{Rfam}, it is easy to see that $R$ is uniquely determined by $\zeta$; specifically,
\begin{align} \label{basturma}
R = i \sum_{a \neq b} P_a \zeta P_b / (E_a - E_b).
\end{align}
Therefore, $R \to \nul$ is equivalent to $\zeta \to \nul$.


Next, by the main assumption of Theorem~\ref{thm:no-go}, all $W$'s and $M_W$'s, and hence all $V_W$'s and $m_W$'s (and thus $Y$), are at least once differentiable with respect to $\rho$. Thus, keeping in mind Eqs.~\eqref{rhofam} and~\eqref{basturma}, we can Taylor-expand $Y(\rho, \omega, H)$ with respect to $\zeta$ around $\zeta = \nul$, which leads us to
\begin{align} \label{yYy}
Y(\rho, \omega, H) = Y(\tau, \omega, H) + i z G(Q, \tau, \omega, H) + o(z),
\end{align}
where $G(Q, \tau, \omega, H)$ is order-$1$ homogeneous in $Q$.

Let us now invoke Constraint~\ref{cons:S1}, which states that the complexness of $Y$ is determined solely by $\omega$ and $\zeta$. Since all the arguments of $Y(\tau, \omega, H)$ and $G(Q, \tau, \omega, H)$ are real matrices (in the eigenbasis of $H$), then Constraint~\ref{cons:S1} enforces that $Y(\tau, \omega, H)$ and $G(Q, \tau, \omega, H)$ are real matrices themselves. Lastly, due to Hermiticity of $Y(\rho, \omega, H)$, $G(Q, \tau, \omega, H)$ must be skew-symmetric. All in all, Constraint~\ref{cons:S1} leads us to
\beaa \label{symmetries}
Y(\tau, \omega, H), \; G(Q, \tau, \omega, H) \; &\in \; \mathbb{R}^{d \times d} \, ;
\\
Y(\tau, \omega, H)^\mathrm{T} &= Y(\tau, \omega, H),
\\
G(Q, \tau, \omega, H)^{\mathrm{T}} &= - G(Q, \tau, \omega, H).
\eeaa
Substituting Eqs.~\eqref{rhofam} and~\eqref{yYy} into Eq.~\eqref{av1ord}, and taking into account that, due to the skew-symmetry of $G(\tau, F, Q, H)$, $\tr\big(\tau \, G(\tau, F, Q, H)\big) = 0$, and that, due to $\diag (\omega) = \nul$, $\tr(\omega \tau) = 0$, we find that
\begin{align} \nonumber
\tr\big(Y(\tau, \omega, H) \tau_{\beta_0}\big) + z \tr\big(\big[Y(\tau, \omega, H) - \omega\big] R \big) = o(z),
\end{align}
which must be satisfied for arbitrary $z \ll 1$. Whence,
\beaa \label{utro}
\tr\big(Y(\tau, \omega, H) \tau\big) &= 0,
\\
\tr\big(\big[Y(\tau, \omega, H) - \omega\big] R \big) &= 0.
\eeaa
Introducing
\begin{align} \nonumber
Y_{\mathrm{d}} = \mathcal{D}_H [Y (\tau, \omega, H)] \quad \mathrm{and} \quad Y_{\mathrm{nd}}= Y (\tau, \omega, H) - Y_{\mathrm{d}},
\end{align}
and keeping in mind that $\mathcal{D}_H[R] = \nul$, we can rewrite Eqs.~\eqref{utro} as
\begin{align} \label{utro1}
\tr\big(Y_{\mathrm{d}} \tau\big) &= 0,
\\ \label{utro2}
\tr\big(\big[Y_{\mathrm{nd}} - \omega\big] R \big) &= 0.
\end{align}
Now, noticing that $Y_{\mathrm{nd}}(\tau, \omega, H) - \omega$ is a real symmetric matrix that is zero in the diagonal subspace of $H$ and independent of $R$, and that Eq.~\eqref{utro2} holds for arbitrary $R$ satisfying the conditions in Eq.~\eqref{Rfam}, we immediately conclude that $Y_{\mathrm{nd}} = \omega$, and thus,
\begin{align} \label{Yaay}
Y(\tau, \omega, H) = Y_{\mathrm{d}} + \omega,
\end{align}
where $Y_{\mathrm{d}}$ is an arbitrary real, symmetric matrix satisfying $Y_{\mathrm{d}} = \mathcal{D}_H[Y_{\mathrm{d}}]$ and Eq.~\eqref{utro1}. We emphasize that this relation was derived exclusively for real $\omega$'s. For more general $\omega$'s, Eq.~\eqref{Yaay} may not apply.

Let us now estimate $\sum_{W \in \mathbf{O}_{\tau, h}} \!\! \beta_0^2 V_W^2 \av{m_W}_\tau$; the need for this will become evident shortly. To do so, let us invoke the operator-valued Cauchy--Schwarz inequality proven in Refs.~\cite{Lance_1995, Jameson_1996}, which will give us
\begin{align} \nonumber
\sum_{W \in \mathbf{O}_{\tau, h}} \!\!\!\! V_W^2 m_W \! &= \! \Big\Vert \! \sum_{W \in \mathbf{O}_{\tau, h}} \!\! (\sqrt{m_W})^2 \Big\Vert \sum_{W \in \mathbf{O}_{\tau, h}} \!\!\! (|V_W| \sqrt{m_W})^2 
\\ \nonumber
&\geq \!\! \bigg[\sum_{W \in \mathbf{O}_{\tau, h}} \! |V_W| m_W \bigg]^2,
\\ \label{Cauchuk}
\end{align}
where in the first line we used the fact that $\sum_{W \in \mathbf{O}_{\rho, h}} m_W = \id$ [see Eq.~\eqref{norm01ord}].

Next, bringing up the quantities $Y^+$ and $Y^-$, and keeping in mind Eq.~\eqref{buys}, we can rewrite Eq.~\eqref{Cauchuk} as
\beaa \label{innneq}
\sum_{W \in \mathbf{O}_{\tau, h}} \!\! V_W^2 m_W &\geq ( Y^+ - Y^- )^2 + 2 \{ Y^+, Y^- \}
\\
&= (Y_{\mathrm{d}} + \omega)^2 + 2 \{ Y^+, Y^- \}
\\
&\geq \omega^2 + \{\omega, Y_{\mathrm{d}}\} + 2 \{ Y^+, Y^- \}. ~~
\eeaa
From here, in view of Constraint~\ref{cons:S2}, we immediately obtain
\begin{align} \label{intermineq}
\sum_{W \in \mathbf{O}_{\tau, h}} \!\! V_W^2 \av{m_W}_\tau \geq \av{\omega^2}_\tau + \av{\{\omega, Y_{\mathrm{d}} \}}_\tau. ~
\end{align}
Now,
\begin{align} \nonumber
\av{\{\omega, Y_{\mathrm{d}} \}}_\tau = \tr(\tau \{\omega, Y_{\mathrm{d}} \}) = \tr(\omega \{\tau, Y_{\mathrm{d}} \}) = 0,
\end{align}
because $\{\tau, Y_{\mathrm{d}} \}$ can be nonzero only in the diagonal subspace of $H$ whereas $\omega$ is zero in exactly that subspace. Thus, for the chosen families of $\omega$'s and $\rho$'s,
\begin{align} \label{ineq1}
\sum_{W \in \mathbf{O}_{\tau, h}} \!\! \beta_0^2 V_W^2 \av{m_W}_\tau \geq \beta_0^2 \av{\omega^2}_\tau .
\end{align}

We are now ready for the final stage for the proof. Let us fix the $\beta_0 > 0$ from Eq.~\eqref{omfam}, and invoke {\reqB}: $\sum_W e^{-\beta_0 W} \tr(\tau M_W) = 1$. As before, noting that this must hold for any $x$, and collecting all powers of $x$, we will see that, on the zeroth- and first-order levels, {\reqB} is secured by Eqs.~\eqref{norm01ord} and~\eqref{av1ord}. Whereas the second-order level produces
\begin{align} \nonumber
\sum_{W \in \mathbf{O}_{\tau, h}} \! \bigg[ \! \bigg( \frac{\beta_0^2 V_W^2}{2} &- \beta_0 v_W \bigg) \av{m_W}_{\tau} - \beta_0 V_W \av{\mu_W}_{\tau}
\\ \nonumber
+& \, \av{\nu_W}_{\tau} \bigg] + \! \sum_{W \in \mathbf{N}_{\tau, h}} e^{-\beta_0 U_W} \av{u_W}_{\tau} = 0,
\end{align}
which, using Eqs.~\eqref{norm2ord} and~\eqref{av2ord}, we can rewrite as
\begin{align} \nonumber
\sum_{W \in \mathbf{O}_{\tau, h}} \!\! \beta_0^2 & V_W^2 \av{m_W}_\tau = \beta_0 \av{i [h, \omega]}_\tau
\\ \nonumber
& - 2 \sum_{W \in \mathbf{N}_{\tau, h}} \!\! \big(e^{-\beta_0 U_W} + \beta_0 U_W - 1 \big) \av{u_W}_\tau,
\end{align}
from where, by noting that $e^{-z} + z - 1 \geq 0$ for $\forall z \in \mathbb{R}$ (with equality only for $z = 0$), we arrive at
\begin{align} \label{jarz2ord2}
\sum_{W \in \mathbf{O}_{\tau, h}} \beta_0^2 V_W^2 \av{m_W}_\tau \leq \beta_0 \av{i [h, \omega]}_\tau.
\end{align}

Lastly, in Subsection~\ref{app:ineq}, we prove that, as long as $\omega \neq \nul$,
\begin{align} \label{ineq2}
\beta_0^2 \av{\omega^2}_\tau > \beta_0 \av{i [h, \omega]}_\tau.
\end{align}
Substituting this into Eq.~\eqref{ineq1}, which was obtained by enforcing {\reqAa} and {\reqAb}, yields
\begin{align} \label{ineq3}
\sum_{W \in \mathbf{O}_{\tau, h}} \! \beta_0^2 V_W^2 \av{m_W}_\tau > \beta_0 \av{i [h, \omega]}_\tau.
\end{align}
We see that this contradicts Eq.~\eqref{jarz2ord2}. Thus, we have just shown that there exist a $\rho$ and a $h$ for which {\reqAa}, {\reqAb}, and {\reqB} cannot hold all at the same time, which concludes the proof of Theorem~\ref{thm:no-go}.

\subsection{Proof of inequality~\eqref{ineq2}}
\label{app:ineq}

To prove inequality~\eqref{ineq2}, let us switch to the eigenbasis of $H$ ordered in such a way that its eigenvalues $E_k$ are ordered increasingly. By explicit calculation, we then see that
\begin{align} \label{blumen}
\beta^2 \av{\omega^2}_\tau = \sum\limits_{k <j} (p_k + p_j) \Delta_{k j}^2 |h_{kj}|^2,
\end{align}
where $p_k = e^{-\beta E_k} / Z$ are the eigenvalues of $\tau$, and $\Delta_{k j} = \beta (E_j - E_k)$.

In the same notation, another simple calculation leads to
\begin{align} \label{flumen}
\beta \av{i [h, \omega]}_\tau = 2 \sum\limits_{k < j} (p_k - p_j) \Delta_{kj} |h_{kj}|^2 \geq 0.
\end{align}
Note that this inequality holds for any passive state $\tau$ and simply expresses the fact that, when $\rho = \tau$ in Eq.~\eqref{eq:avWexpand}, the term $\propto x$ vanishes, and the second term has to be positive because $\av{W}$ is the work performed on a passive state and thus has to be nonnegative.

Subtracting~\eqref{flumen} from~\eqref{blumen}, we obtain
\begin{align} \nonumber
\beta^2 & \av{\omega^2}_\tau - \beta \av{i [h, \omega]}_\tau
\\ \label{filone}
&= 2 \sum_{k < j} |h_{kj}|^2 (p_k + p_j) \Delta_{kj} \Big[ \frac{\Delta_{kj}}{2} - \tanh \frac{\Delta_{kj}}{2} \Big]. ~~
\end{align}
Now, $[h, H] \neq \nul$ means that $h$ has nonzero nondiagonal elements between some eigensubspaces of $H$. In other words, there exist some $k_0$ and $j_0$ for which $\Delta_{k_0 j_0} > 0$ and $h_{k_0 j_0} \neq 0$. Moreover, since $x - \tanh x > 0$ whenever $x > 0$, at least one of the summands in the right-hand side of Eq~\eqref{filone} is $> 0$, which thereby proves the bound in Eq.~\eqref{ineq2}.

\section{A scheme simultaneously satisfying \textcolor{red}{$\mathfrak{A}_1$} and \textcolor{red}{$\mathfrak{B}$}}
\label{app:scheme}

Let us construct a scheme that satisfies {\reqAa} and {\reqB} and depends continuously on the system's state. The example illustrates that {\reqAb} is an independent and essential requirement for any meaningful measurement of work.

As a first step, we devise a measurement scheme in the specific class of processes in which the system starts out in a possibly coherent (in the eigenbasis of $H$) state and unitarily evolves into a diagonal state: $(H, \rho) \longrightarrow (H', \rho_D)$, where $\rho_D = V \rho V^\dagger$ and $[\rho_D, H'] = \nul$. Here, we notice that the time-reversed process, $(H', \rho_D) \longrightarrow (H, \rho)$, with $\rho = V^\dagger \rho_D V$, is describable by the TPM scheme. 

Now, if the work statistics of a unitary process is $\{ (W_\alpha, p_\alpha) \}_\alpha$, then, to the time-reversed process, it is reasonable to prescribe the work statistics $\{ (-W_\alpha, p_\alpha) \}_\alpha$.

With this prescription, thinking of the $\mathcal{P} = \{ (H, \rho) \overset{V}{\longrightarrow} (H', \rho_D) \}$ process as the reverse of $\widetilde{\mathcal{P}} = \{ (H, \rho_D) \overset{V^\dagger}{\longrightarrow} (H', \rho) \}$, and describing the latter by the TPM scheme [see Eqs.~\eqref{tpem} and~\eqref{salame}], for the work statistics of $\mathcal{P}$ we obtain
\begin{align} \nonumber
W_{a k} &= - (E_a - E'_k),
\\ \nonumber
p_{a k} &= \tr \left( P'_k \rho_D \right) \tr \left( P_a \frac{V^\dagger P'_k \rho_D P'_k V}{\tr \left( P'_k \rho_D \right)} \right)
\\ \nonumber
&= \tr \left( P_a V^\dagger P'_k \rho_D P'_k V \right).
\end{align}
Put in other words:
\begin{align} \label{prot1}
W_{a k} = E'_k - E_a, \qquad p_{a k} = \tr(M_{a k} \rho),
\end{align}
with
\begin{align} \label{prot2}
M_{a k} = V^\dagger P'_k V P_a V^\dagger P'_k V.
\end{align}
Note that $M_{a k}$ depends on the initial state through $V$. We call this measurement ``reverse'' TPM scheme as opposed to the standard TPM scheme that directly measures the ``forward'' process. Interestingly, the work distribution produced by the reverse TPM scheme coincides with that of the scheme proposed in Ref.~\cite{Allahverdyan_2014}. Extraction of ergotropy \cite{Allahverdyan_2004} is an important example of a process where the system goes from a possibly coherent state to a diagonal one (the passive state \cite{Pusz_1978, Lenard_1978}), and the reverse TPM scheme is thus a natural choice for measuring ergotropy fluctuations.

Having at our disposal this scheme, we can now describe any coherent-to-coherent process by decomposing it into coherent-to-incoherent and incoherent-to-coherent processes. More specifically, consider an arbitrary Hamiltonian process $(H, \rho) \longrightarrow (H', \rho')$, with $\rho' = U \rho U^\dagger$, for which both $[\rho, H] \neq \nul$ and $[\rho', H'] \neq \nul$, and decompose it into
\begin{align} \label{proc}
(H, \rho) \overset{R}{\xrightarrow{\hspace*{6 mm}}} (H, \widetilde{\rho}) \overset{UR^\dagger}{\xrightarrow{\hspace*{6 mm}}} (H', \rho'), ~
\end{align}
where $\widetilde{\rho} = R \rho R^\dagger$ is diagonal in $H$. Moreover, $R$ is chosen such that the diagonal of $\widetilde{\rho}$ is ordered in the same way as the diagonal of $\rho$. More specifically, say, $\hat{\pi}$ is the permutation matrix that reorders $\rho_D^{\downarrow}$ into $\rho_D$: $\rho_D = \hat{\pi} \rho_D^\downarrow \hat{\pi}^{\mathrm{T}}$, where $\rho_D$ is a diagonal matrix whose diagonal coincides with that of $\rho$ in the eigenbasis of $H$. Then, if $r^{\downarrow}$ is the vector composed of the eigenvalues of $\rho$, organized in the decreasing order, then $\widetilde{\rho} = \hat{\pi} r^\downarrow \hat{\pi}^{\mathrm{T}}$. Thus, $R$ is the unitary operator that rotates $\rho$ into $\widetilde{\rho}$. The operator $R$ is chosen in this way to guarantee that, when $\rho$ is diagonal, then $R = \id$.

Now, using Eqs.~\eqref{prot1} and~\eqref{prot2}, for the first part of the process~\eqref{proc}, we obtain
\begin{align} \label{lavash}
W^{\RN{1}}_{a b} = E_b - E_a, \qquad p^{\RN{1}}_{a b} = \tr(M^{\RN{1}}_{a b} \rho),
\end{align}
where the POVM is given by
\begin{align} \label{tavish}
M^{\RN{1}}_{a b} = R^\dagger P_b R P_a R^\dagger P_b R.
\end{align}
Following the standard TPM scheme, for the second part of the process, we find
\begin{align} \label{matnaqash}
W^{\RN{2}}_{c k} = E'_k - E_c, \qquad p^{\RN{2}}_{c k} = \tr(M^{\RN{2}}_{c k} \rho),
\end{align}
where
\begin{align} \label{boqon}
M^{\RN{2}}_{c k} = R^\dagger P_c R U^\dagger P_k' U R^\dagger P_c R.
\end{align}
Finally, for the overall process~\eqref{proc}, let us write
\beaa \label{erebuni}
W_{a b c k} &= W^{\RN{1}}_{a b} + W^{\RN{2}}_{c k},
\\
p_{a b c k} &= p^{\RN{1}}_{a b} p^{\RN{2}}_{c k} = \tr\left( M^{\RN{1}}_{a b} \otimes M^{\RN{2}}_{c k} \cdot \rho \otimes \rho \right),
\eeaa
where $M_{a b c k} = M^{\RN{1}}_{a b} \otimes M^{\RN{2}}_{c k}$ is a POVM that depends on the initial state through $R$. 

From the perspective of a single copy of the system, this scheme is equivalent to the following state-dependent POVM. With the same outcomes $W_{a b c k}$, $p_{a b c k}$ can be written as
\begin{align}
p_{a b c k} = \tr\big[\widetilde{M}(\rho)_{a b c k} \, \rho\big],
\end{align}
where
\begin{align}
\widetilde{M}(\rho)_{a b c k} = p^{\RN{2}}_{c k} M^{\RN{1}}_{a b}.
\end{align}
Obviously, the operators $\widetilde{M}(\rho)_{a b c k}$ constitute a POVM.

By construction, this definition satisfies {\reqAa} and, keeping in mind that $R = \id$ for initially diagonal states, it coincides with the TPM scheme, and therefore satisfies {\reqB}. However, for $[\rho, H] \neq \nul$, this definition does not satisfy {\reqAb}. Indeed, in order to measure work for some $[\rho, H] \neq \nul$ in the trivial case of $H' = H$ and $U = \id$, the scheme will first apply a unitary to diagonalize $\rho$ with $H$, thereby producing some nontrivial statistics according to Eqs.~\eqref{lavash} and~\eqref{tavish}; then, it will rotate that state back to $\rho$, again producing nontrivial statistics as per Eqs.~\eqref{matnaqash} and~\eqref{boqon}. In total, for the trivial process, the scheme will produce rather complicated work statistics described by Eq.~\eqref{erebuni}.

\section{The other operator of work}
\label{app:rotateW}

The operator $\Upsilon_\rho$ is similar to the operator for work proposed in Ref.~\cite{Yukawa_2000}. There, the work operator was introduced according to $e^{-\beta \hat{W}} = e^{-\beta U^\dagger H' U} e^{\beta H}$. However, the work operator defined in this way will not be Hermitian whenever $[U^\dagger H' U, H] \neq \nul$. Our operator $\Upsilon_\rho$ is Hermitian by construction and is defined for all states.

First of all, let us show that this scheme is covered by Theorem~\ref{thm:no-go}. It is self-evident that $\hat{\beta}(\rho)$, and therefore $\Upsilon_\rho$, are differentiable functions of $\rho$, which means that both the measurement operators and outcomes of this scheme are differentiable with respect to $\rho$. Moreover, $\Upsilon_\rho$ satisfies Constraint~\ref{cons:S1}, because, for processes~\eqref{processes}, $U^\dagger H U = H + x \omega + O(x^2)$, and therefore, except for $\omega$, all the operators in the expansion of the relevant quantities with respect to $x$ will be operator functions of $H$. Consequently, they will all be real in the eigenbasis of $H$, for all values of $\rho$, which means that the dependence of complexness on $\rho$ is trivial. As regards Constraint~\ref{cons:S2}, it is satisfied by construction: $\Upsilon_\rho$ is an operator scheme, which means that all $M_W$, and thus all $m_W$ are eigenprojectors of a Hermitian operator, and are thus mutually orthogonal, which in turn means that $Y^+ Y^- = \nul$ [see also the discussion under Eq.~\eqref{huys}].

Since $\Upsilon_\rho$ satisfies {\reqAb} and {\reqB} by construction, the fact that Theorem~\ref{thm:no-go} applies to this scheme means that it necessarily has to violate {\reqAa}.

To explicitly show that this is indeed the case, let us consider the processes $\Pi$, defined in Eq.~\eqref{processes}, and thermal initial states $\tau \propto e^{-\beta H}$. First, we rewrite Eq.~\eqref{rotow} as
\begin{align} \label{rollo}
\Upsilon_\tau = - \beta^{-1} \ln \big( \tau^{-1/2} U^\dagger \tau U \tau^{-1/2} \big).
\end{align}
Through a simple calculation employing the Baker--Hausdorff lemma and the Taylor expansion of $\ln(1 + x)$ around $x = 0$, Eq.~\eqref{rollo} will lead us to
\beaa \nonumber
-\beta \Upsilon_\tau =& \, i x \tau^{-1/2} [h, \tau] \tau^{-1/2}
\\
&- \frac{x^2}{2} \tau^{-1/2} [h, [h, \tau]] \tau^{-1/2}
\\
&+ \frac{x^2}{2} \tau^{-1/2} [h, \tau] \tau^{-1}[h, \tau] \tau^{-1/2} + O(x^3).
\eeaa
Therefore,
\begin{align} \label{gazon1}
\beta \av{\Upsilon_\tau}_\tau = - \frac{x^2}{2} \tr(\tau^{-1} [h, \tau]^2) + O(x^3).
\end{align}
Whereas, reading from Eq.~\eqref{eq:avWexpand}, we have
\begin{align} \label{gazon2}
\beta \av{\Omega}_\tau = \frac{x^2}{2} \tr([h, \tau] \, [h, \beta H]) + O(x^3).
\end{align}
Neglecting the $O(x^3)$ terms and working in the eigenbasis of $H$, we can rewrite Eqs.~\eqref{gazon1} and~\eqref{gazon2} as
\begin{align} \label{pixa1}
\beta \av{\Upsilon_\tau}_\tau &= x^2 \sum_{k < j} |h_{k j}|^2 (p_k - p_j)^2 \frac{p_k + p_j}{2 p_k p_j} 
\\ \label{pixa2}
&= x^2 \sum_{k < j} |h_{k j}|^2 p_j (e^{\Delta_{k j}} - 1) \sinh \Delta_{k j}
\end{align}
and
\begin{align} \label{pixa3}
\beta \av{\Omega}_\tau &= x^2 \sum_{k < j} |h_{k j}|^2 (p_k - p_j) \Delta_{k j}
\\ \label{pixa4}
&= x^2 \sum_{k < j} |h_{k j}|^2 p_j (e^{\Delta_{k j}} - 1) \Delta_{k j},
\end{align}
where $p_k = e^{-\beta E_k} / Z$ are the eigenvalues of $\tau$, and $\Delta_{k j} = \beta (E_j - E_k)$; when transitioning from~\eqref{pixa1} to~\eqref{pixa2} and from~\eqref{pixa3} to~\eqref{pixa4}, we used that $p_k = p_j e^{\Delta_{k j}}$.

Now, whenever $[h, H] \neq \nul$, there exist such $k_0$ and $j_0$ for which $\Delta_{k_0 j_0} > 0$ and $h_{k_0 j_0} \neq 0$. Since $\Delta_{k j} \geq 0$ for all $k < j$, Eqs.~\eqref{pixa2} and~\eqref{pixa4} thus imply that
\begin{align}
\av{\Upsilon_\tau}_\tau > \av{\Omega}_\tau,
\end{align}
provided $x$ is small enough for the $O(x^3)$ terms to be irrelevant. This inequality simply means that $\Upsilon_\tau$ violates {\reqAa} for all thermal states $\tau$ as long as $[h, H] \neq \nul$.

\subsection{Two curious properties of $\Upsilon_\rho$}

Below, we will prove two results about $\Upsilon_\rho$ that are not related to the main goal of this appendix. However, we feel compelled to present them to more completely characterize the yet another ``definition of work'' that $\Upsilon_\rho$ represents.
\begin{proposition} \label{thm:Yuk1}
\begin{align} \label{eq:WMineq}
\av{\Omega}_{\tau} - \av{\Upsilon_\tau}_{\tau} \leq \beta^{-1} S(\tau \, \Vert \, U^\dagger \tau' U);
\end{align}
\end{proposition}
\begin{proposition} \label{thm:Yuk2}
\begin{align} \label{eq:WMmaj}
\mathrm{Spec}(\Upsilon_\tau)^{\downarrow} \succ \mathrm{Spec}(\Omega)^\downarrow,
\end{align}
for any $\beta$.
\end{proposition}
Here, $\tau' = e^{-\beta H'}/ \tr e^{-\beta H'}$ and $\mathrm{Spec}(O)$ is the spectrum of the operator $O$. The superscript $\downarrow$ indicates that the set is ordered from large to small, and $\succ$ is the majorization relation \cite{hojo}.

\begin{proof}[Proof of Proposition~\ref{thm:Yuk1}]
Let us first observe that, as a trivial consequence of the operator-concavity of the logarithm,
\begin{align}
\tr(\rho \ln O) \leq \ln\tr(\rho O),
\end{align}
for any density operator $\rho$ and Hermitian operator $O \geq \nul$. Furthermore, noticing that $e^{\beta H/2}e^{-\beta U^\dagger H' U} e^{\beta H/2} > \nul$, we obtain
\begin{align} \label{zartuxi}
\av{\Upsilon_\tau}_\tau = - \beta^{-1} \tr \Big( \tau \ln e^{- \beta \small{\Upsilon_\tau}} \Big) \geq \Delta F,
\end{align}
where we took into account that, by construction,
\begin{align}
\big\langle e^{-\beta \small{\Upsilon_\tau}} \big\rangle_\tau = e^{-\beta \Delta F} = \frac{Z_\beta[H']}{Z_\beta[H]}.
\end{align}
Interestingly, inequality~\eqref{zartuxi} provides another aspect by which $\Upsilon_\tau$ resembles work.

Finally, observing that
\begin{align}
\Delta F = \av{\Omega}_\tau - \beta^{-1} S(\tau \, \Vert \, U^\dagger \tau' U),
\end{align}
we justify Eq.~\eqref{eq:WMineq}.
\end{proof}

\begin{proof}[Proof of Proposition~\ref{thm:Yuk2}]
For this proof, we invoke the following two theorems, proven in Refs.~\cite{Lenard_1971, Thompson_1971}, that hold for arbitrary Hermitian operators $A$ and $B$.
\begin{theorem}[The Theorem of Ref.~\cite{Lenard_1971}] \label{thm:lenard}
An arbitrary neighborhood of $e^{A + B}$ contains an $X$ for which there exist two unitary operators $\mathcal{U}$ and $\mathcal{V}$ such that
\begin{align}
\mathrm{Spec}(X) = \mathrm{Spec}\big( Y^{\frac{1}{2}} \mathcal{U} Y^{\frac{1}{2}} \mathcal{V} \big),
\end{align}
where
\begin{align}
Y = e^{A/2} e^B e^{A/2}.
\end{align}
\end{theorem}
\begin{theorem}[Theorem 1 of Ref.~\cite{Thompson_1971}] \label{thm:thompson}
If $X_1 \geq X_2 \geq \cdots \geq X_d \geq 0$ and $Y_1 \geq Y_2 \geq \cdots \geq Y_d \geq 0$ are the eigenvalues of positive-semidefinite Hermitian matrices $X$ and $Y$, then unitary operators $\mathcal{U}$ and $\mathcal{V}$ exist such that $\mathrm{Spec}(X) = \mathrm{Spec}\big( Y^{\frac{1}{2}} \mathcal{U} Y^{\frac{1}{2}} \mathcal{V} \big)$ if and only if
\begin{align} \label{eq:Thompson_1}
X_1 X_2 \cdots X_k \leq Y_1 Y_2 \cdots Y_k,
\end{align}
for $k = 1, 2, ..., d-1$, and
\begin{align} \label{eq:Thompson_2}
\det X = \det Y.
\end{align}
\end{theorem}

In our situation, $A = \beta H$ and $B = -\beta U^\dagger H' U$, and thus $e^{A + B} = e^{-\beta \Omega}$ and $Y = e^{-\beta \Upsilon_\tau}$. Therefore, Theorems~\ref{thm:lenard} and~\ref{thm:thompson} imply that, given some matrix norm $\Vert \cdot \Vert$ (e.g., the operator norm), for any $\epsilon > 0$, there exists an $\Omega^{(\epsilon)}$ such that $\Vert \Omega - \Omega^{(\epsilon)} \Vert < \epsilon$ and, according to Eq.~\eqref{eq:Thompson_1},
\begin{align} \label{eq:majmuj}
\Omega^{(\epsilon)}_d + \cdots + \Omega^{(\epsilon)}_{d - k + 1} \geq \Gamma_d + \cdots + \Gamma_{d - k + 1},
\end{align}
for $k = 1, 2, ..., d-1$, where $\Omega^{(\epsilon)}_1 \geq \cdots \geq \Omega^{(\epsilon)}_d$ are the ordered eigenvalues of $\Omega^{(\epsilon)}$ and $\Gamma_1 \geq \cdots \geq \Gamma_d$ are the ordered eigenvalues of $\Upsilon_\tau$. Eq.~\eqref{eq:Thompson_2}, in turn, entails
\begin{align} \label{eq:zarbab}
\sum_{k=1}^d \Omega^{(\epsilon)}_k = \tr \Omega^{(\epsilon)} = \tr \Upsilon_\tau = \sum_{k=1}^d \Gamma_k.
\end{align}
Taking the limit $\epsilon \to 0$ in Eqs.~\eqref{eq:majmuj} and~\eqref{eq:zarbab}, it is straightforward to see that $\sum_{m=1}^k \Omega_m \leq \sum_{m=1}^k \Gamma_m$, for $k = 1,..., d-1$ and $\tr \Omega = \tr \Upsilon_\tau$, which concludes the proof of Proposition~\ref{thm:Yuk2}.
\end{proof}

\bibliographystyle{sirun}
\bibliography{references}

\begin{thebibliography}{77}%
\makeatletter
\providecommand \@ifxundefined [1]{%
 \@ifx{#1\undefined}
}%
\providecommand \@ifnum [1]{%
 \ifnum #1\expandafter \@firstoftwo
 \else \expandafter \@secondoftwo
 \fi
}%
\providecommand \@ifx [1]{%
 \ifx #1\expandafter \@firstoftwo
 \else \expandafter \@secondoftwo
 \fi
}%
\providecommand \natexlab [1]{#1}%
\providecommand \enquote  [1]{``#1''}%
\providecommand \bibnamefont  [1]{#1}%
\providecommand \bibfnamefont [1]{#1}%
\providecommand \citenamefont [1]{#1}%
\providecommand \href@noop [0]{\@secondoftwo}%
\providecommand \href [0]{\begingroup \@sanitize@url \@href}%
\providecommand \@href[1]{\@@startlink{#1}\@@href}%
\providecommand \@@href[1]{\endgroup#1\@@endlink}%
\providecommand \@sanitize@url [0]{\catcode `\\12\catcode `\$12\catcode
  `\&12\catcode `\#12\catcode `\^12\catcode `\_12\catcode `\%12\relax}%
\providecommand \@@startlink[1]{}%
\providecommand \@@endlink[0]{}%
\providecommand \url  [0]{\begingroup\@sanitize@url \@url }%
\providecommand \@url [1]{\endgroup\@href {#1}{\urlprefix }}%
\providecommand \urlprefix  [0]{URL }%
\providecommand \Eprint [0]{\href }%
\providecommand \doibase [0]{http://doi.org/}%
\providecommand \selectlanguage [0]{\@gobble}%
\providecommand \bibinfo  [0]{\@secondoftwo}%
\providecommand \bibfield  [0]{\@secondoftwo}%
\providecommand \translation [1]{[#1]}%
\providecommand \BibitemOpen [0]{}%
\providecommand \bibitemStop [0]{}%
\providecommand \bibitemNoStop [0]{.\EOS\space}%
\providecommand \EOS [0]{\spacefactor3000\relax}%
\providecommand \BibitemShut  [1]{\csname bibitem#1\endcsname}%
\let\auto@bib@innerbib\@empty
\bibitem [{\citenamefont {Bochkov}\ and\ \citenamefont
  {Kuzovlev}(1977)}]{Bochkov_1977}%
  \BibitemOpen
  \bibfield  {author} {\bibinfo {author} {\bibfnamefont {G.~N.}\ \bibnamefont
  {Bochkov}}\ and\ \bibinfo {author} {\bibfnamefont {Y.~B.}\ \bibnamefont
  {Kuzovlev}},\ }\bibfield  {title} {\emph {\bibinfo {title} {General theory of
  thermal fluctuations in nonlinear systems},\ }}\href
  {http://www.jetp.ras.ru/cgi-bin/r/index/e/45/1/p125?a=list} {\bibfield
  {journal} {\bibinfo  {journal} {Sov. Phys. JETP}\ }\textbf {\bibinfo {volume}
  {45}},\ \bibinfo {pages} {125} (\bibinfo {year} {1977})}\BibitemShut
  {NoStop}%
\bibitem [{\citenamefont {Jarzynski}(1997)}]{Jarzynski_1997}%
  \BibitemOpen
  \bibfield  {author} {\bibinfo {author} {\bibfnamefont {C.}~\bibnamefont
  {Jarzynski}},\ }\bibfield  {title} {\emph {\bibinfo {title} {Nonequilibrium
  equality for free energy differences},\ }}\href
  {\doibase10.1103/PhysRevLett.78.2690} {\bibfield  {journal} {\bibinfo
  {journal} {Phys. Rev. Lett.}\ }\textbf {\bibinfo {volume} {78}},\ \bibinfo
  {pages} {2690} (\bibinfo {year} {1997})}\BibitemShut {NoStop}%
\bibitem [{\citenamefont {Kurchan}()}]{Kurchan_2000}%
  \BibitemOpen
  \bibfield  {author} {\bibinfo {author} {\bibfnamefont {J.}~\bibnamefont
  {Kurchan}},\ }\bibfield  {title} {\emph {\bibinfo {title} {A quantum
  fluctuation theorem},\ }}\href@noop {} {\ }\Eprint
  {https://doi.org/10.48550/arXiv.cond-mat/0007360}{arXiv:cond-mat/0007360}\BibitemShut
  {NoStop}%
\bibitem [{\citenamefont {Tasaki}()}]{Tasaki_2000}%
  \BibitemOpen
  \bibfield  {author} {\bibinfo {author} {\bibfnamefont {H.}~\bibnamefont
  {Tasaki}},\ }\bibfield  {title} {\emph {\bibinfo {title} {Jarzynski relations
  for quantum systems and some applications},\ }}\href@noop {} {\ }\Eprint
  {https://doi.org/10.48550/arXiv.cond-mat/0009244}{arXiv:cond-mat/0009244}\BibitemShut
  {NoStop}%
\bibitem [{\citenamefont {Allahverdyan}\ and\ \citenamefont
  {Nieuwenhuizen}(2005)}]{Allahverdyan_2005}%
  \BibitemOpen
  \bibfield  {author} {\bibinfo {author} {\bibfnamefont {A.~E.}\ \bibnamefont
  {Allahverdyan}}\ and\ \bibinfo {author} {\bibfnamefont {T.~M.}\ \bibnamefont
  {Nieuwenhuizen}},\ }\bibfield  {title} {\emph {\bibinfo {title} {Fluctuations
  of work from quantum subensembles: {T}he case against quantum
  work-fluctuation theorems},\ }}\href {\doibase10.1103/PhysRevE.71.066102}
  {\bibfield  {journal} {\bibinfo  {journal} {Phys. Rev. E}\ }\textbf {\bibinfo
  {volume} {71}},\ \bibinfo {pages} {066102} (\bibinfo {year}
  {2005})}\BibitemShut {NoStop}%
\bibitem [{\citenamefont {Esposito}\ \emph {et~al.}(2009)\citenamefont
  {Esposito}, \citenamefont {Harbola},\ and\ \citenamefont
  {Mukamel}}]{Esposito_2009}%
  \BibitemOpen
  \bibfield  {author} {\bibinfo {author} {\bibfnamefont {M.}~\bibnamefont
  {Esposito}}, \bibinfo {author} {\bibfnamefont {U.}~\bibnamefont {Harbola}}, \
  and\ \bibinfo {author} {\bibfnamefont {S.}~\bibnamefont {Mukamel}},\
  }\bibfield  {title} {\emph {\bibinfo {title} {Nonequilibrium fluctuations,
  fluctuation theorems, and counting statistics in quantum systems},\ }}\href
  {\doibase10.1103/RevModPhys.81.1665} {\bibfield  {journal} {\bibinfo
  {journal} {Rev. Mod. Phys.}\ }\textbf {\bibinfo {volume} {81}},\ \bibinfo
  {pages} {1665} (\bibinfo {year} {2009})}\BibitemShut {NoStop}%
\bibitem [{\citenamefont {Allahverdyan}(2014)}]{Allahverdyan_2014}%
  \BibitemOpen
  \bibfield  {author} {\bibinfo {author} {\bibfnamefont {A.~E.}\ \bibnamefont
  {Allahverdyan}},\ }\bibfield  {title} {\emph {\bibinfo {title}
  {Nonequilibrium quantum fluctuations of work},\ }}\href
  {\doibase10.1103/PhysRevE.90.032137} {\bibfield  {journal} {\bibinfo
  {journal} {Phys. Rev. E}\ }\textbf {\bibinfo {volume} {90}},\ \bibinfo
  {pages} {032137} (\bibinfo {year} {2014})}\BibitemShut {NoStop}%
\bibitem [{\citenamefont {Talkner}\ and\ \citenamefont
  {H\"{a}nggi}(2016)}]{Talkner_2016}%
  \BibitemOpen
  \bibfield  {author} {\bibinfo {author} {\bibfnamefont {P.}~\bibnamefont
  {Talkner}}\ and\ \bibinfo {author} {\bibfnamefont {P.}~\bibnamefont
  {H\"{a}nggi}},\ }\bibfield  {title} {\emph {\bibinfo {title} {Aspects of
  quantum work},\ }}\href {\doibase10.1103/PhysRevE.93.022131} {\bibfield
  {journal} {\bibinfo  {journal} {Phys. Rev. E}\ }\textbf {\bibinfo {volume}
  {93}},\ \bibinfo {pages} {022131} (\bibinfo {year} {2016})}\BibitemShut
  {NoStop}%
\bibitem [{\citenamefont {Perarnau-Llobet}\ \emph {et~al.}(2017)\citenamefont
  {Perarnau-Llobet}, \citenamefont {B\"{a}umer}, \citenamefont {Hovhannisyan},
  \citenamefont {Huber},\ and\ \citenamefont
  {Ac\'{i}n}}]{Perarnau-Llobet_PRL_2017}%
  \BibitemOpen
  \bibfield  {author} {\bibinfo {author} {\bibfnamefont {M.}~\bibnamefont
  {Perarnau-Llobet}}, \bibinfo {author} {\bibfnamefont {E.}~\bibnamefont
  {B\"{a}umer}}, \bibinfo {author} {\bibfnamefont {K.~V.}\ \bibnamefont
  {Hovhannisyan}}, \bibinfo {author} {\bibfnamefont {M.}~\bibnamefont {Huber}},
  \ and\ \bibinfo {author} {\bibfnamefont {A.}~\bibnamefont {Ac\'{i}n}},\
  }\bibfield  {title} {\emph {\bibinfo {title} {No-go theorem for the
  characterization of work fluctuations in coherent quantum systems},\ }}\href
  {\doibase10.1103/PhysRevLett.118.070601} {\bibfield  {journal} {\bibinfo
  {journal} {Phys. Rev. Lett.}\ }\textbf {\bibinfo {volume} {118}},\ \bibinfo
  {pages} {070601} (\bibinfo {year} {2017})}\BibitemShut {NoStop}%
\bibitem [{\citenamefont {Sampaio}\ \emph {et~al.}(2018)\citenamefont
  {Sampaio}, \citenamefont {Suomela}, \citenamefont {Ala-Nissila},
  \citenamefont {Anders},\ and\ \citenamefont {Philbin}}]{Sampaio_2018}%
  \BibitemOpen
  \bibfield  {author} {\bibinfo {author} {\bibfnamefont {R.}~\bibnamefont
  {Sampaio}}, \bibinfo {author} {\bibfnamefont {S.}~\bibnamefont {Suomela}},
  \bibinfo {author} {\bibfnamefont {T.}~\bibnamefont {Ala-Nissila}}, \bibinfo
  {author} {\bibfnamefont {J.}~\bibnamefont {Anders}}, \ and\ \bibinfo {author}
  {\bibfnamefont {T.~G.}\ \bibnamefont {Philbin}},\ }\bibfield  {title} {\emph
  {\bibinfo {title} {Quantum work in the {B}ohmian framework},\ }}\href
  {\doibase10.1103/PhysRevA.97.012131} {\bibfield  {journal} {\bibinfo
  {journal} {Phys. Rev. A}\ }\textbf {\bibinfo {volume} {97}},\ \bibinfo
  {pages} {012131} (\bibinfo {year} {2018})}\BibitemShut {NoStop}%
\bibitem [{\citenamefont {B\"{a}umer}\ \emph {et~al.}(2018)\citenamefont
  {B\"{a}umer}, \citenamefont {Lostaglio}, \citenamefont {Perarnau-Llobet},\
  and\ \citenamefont {Sampaio}}]{Baumer_2018}%
  \BibitemOpen
  \bibfield  {author} {\bibinfo {author} {\bibfnamefont {E.}~\bibnamefont
  {B\"{a}umer}}, \bibinfo {author} {\bibfnamefont {M.}~\bibnamefont
  {Lostaglio}}, \bibinfo {author} {\bibfnamefont {M.}~\bibnamefont
  {Perarnau-Llobet}}, \ and\ \bibinfo {author} {\bibfnamefont {R.}~\bibnamefont
  {Sampaio}},\ }\bibinfo {title} {Fluctuating work in coherent quantum systems:
  Proposals and limitations},\ in\ \href {\doibase10.1007/978-3-319-99046-0_11}
  {\emph {\bibinfo {booktitle} {Thermodynamics in the Quantum Regime:
  Fundamental Aspects and New Directions}}},\ \bibinfo {editor} {edited by\
  \bibinfo {editor} {\bibfnamefont {F.}~\bibnamefont {Binder}}, \bibinfo
  {editor} {\bibfnamefont {L.~A.}\ \bibnamefont {Correa}}, \bibinfo {editor}
  {\bibfnamefont {C.}~\bibnamefont {Gogolin}}, \bibinfo {editor} {\bibfnamefont
  {J.}~\bibnamefont {Anders}}, \ and\ \bibinfo {editor} {\bibfnamefont
  {G.}~\bibnamefont {Adesso}}}\ (\bibinfo  {publisher} {Springer International
  Publishing},\ \bibinfo {address} {Cham},\ \bibinfo {year} {2018})\ pp.\
  \bibinfo {pages} {275--300}\BibitemShut {NoStop}%
\bibitem [{\citenamefont {Brodier}\ \emph {et~al.}(2020)\citenamefont
  {Brodier}, \citenamefont {Mallick},\ and\ \citenamefont {Ozorio~de
  Almeida}}]{Brodier_2020}%
  \BibitemOpen
  \bibfield  {author} {\bibinfo {author} {\bibfnamefont {O.}~\bibnamefont
  {Brodier}}, \bibinfo {author} {\bibfnamefont {K.}~\bibnamefont {Mallick}}, \
  and\ \bibinfo {author} {\bibfnamefont {A.~M.}\ \bibnamefont {Ozorio~de
  Almeida}},\ }\bibfield  {title} {\emph {\bibinfo {title} {Semiclassical work
  and quantum work identities in {W}eyl representation},\ }}\href
  {\doibase10.1088/1751-8121/ab8110} {\bibfield  {journal} {\bibinfo  {journal}
  {J. Phys. A}\ }\textbf {\bibinfo {volume} {53}},\ \bibinfo {pages} {325001}
  (\bibinfo {year} {2020})}\BibitemShut {NoStop}%
\bibitem [{\citenamefont {Yukawa}(2000)}]{Yukawa_2000}%
  \BibitemOpen
  \bibfield  {author} {\bibinfo {author} {\bibfnamefont {S.}~\bibnamefont
  {Yukawa}},\ }\bibfield  {title} {\emph {\bibinfo {title} {A quantum analogue
  of the {J}arzynski equality},\ }}\href {\doibase10.1143/JPSJ.69.2367}
  {\bibfield  {journal} {\bibinfo  {journal} {J. Phys. Soc. Jpn.}\ }\textbf
  {\bibinfo {volume} {69}},\ \bibinfo {pages} {2367} (\bibinfo {year}
  {2000})}\BibitemShut {NoStop}%
\bibitem [{\citenamefont {Talkner}\ \emph {et~al.}(2007)\citenamefont
  {Talkner}, \citenamefont {Lutz},\ and\ \citenamefont
  {H\"{a}nggi}}]{Talkner_2007}%
  \BibitemOpen
  \bibfield  {author} {\bibinfo {author} {\bibfnamefont {P.}~\bibnamefont
  {Talkner}}, \bibinfo {author} {\bibfnamefont {E.}~\bibnamefont {Lutz}}, \
  and\ \bibinfo {author} {\bibfnamefont {P.}~\bibnamefont {H\"{a}nggi}},\
  }\bibfield  {title} {\emph {\bibinfo {title} {Fluctuation theorems: {W}ork is
  not an observable},\ }}\href {\doibase10.1103/PhysRevE.75.050102} {\bibfield
  {journal} {\bibinfo  {journal} {Phys. Rev. E}\ }\textbf {\bibinfo {volume}
  {75}},\ \bibinfo {pages} {050102(R)} (\bibinfo {year} {2007})}\BibitemShut
  {NoStop}%
\bibitem [{\citenamefont {Solinas}\ and\ \citenamefont
  {Gasparinetti}(2015)}]{Solinas_2015}%
  \BibitemOpen
  \bibfield  {author} {\bibinfo {author} {\bibfnamefont {P.}~\bibnamefont
  {Solinas}}\ and\ \bibinfo {author} {\bibfnamefont {S.}~\bibnamefont
  {Gasparinetti}},\ }\bibfield  {title} {\emph {\bibinfo {title} {Full
  distribution of work done on a quantum system for arbitrary initial states},\
  }}\href {\doibase10.1103/PhysRevE.92.042150} {\bibfield  {journal} {\bibinfo
  {journal} {Phys. Rev. E}\ }\textbf {\bibinfo {volume} {92}},\ \bibinfo
  {pages} {042150} (\bibinfo {year} {2015})}\BibitemShut {NoStop}%
\bibitem [{\citenamefont {Deffner}\ \emph {et~al.}(2016)\citenamefont
  {Deffner}, \citenamefont {Paz},\ and\ \citenamefont {Zurek}}]{Deffner_2016}%
  \BibitemOpen
  \bibfield  {author} {\bibinfo {author} {\bibfnamefont {S.}~\bibnamefont
  {Deffner}}, \bibinfo {author} {\bibfnamefont {J.~P.}\ \bibnamefont {Paz}}, \
  and\ \bibinfo {author} {\bibfnamefont {W.~H.}\ \bibnamefont {Zurek}},\
  }\bibfield  {title} {\emph {\bibinfo {title} {Quantum work and the
  thermodynamic cost of quantum measurements},\ }}\href
  {\doibase10.1103/PhysRevE.94.010103} {\bibfield  {journal} {\bibinfo
  {journal} {Phys. Rev. E}\ }\textbf {\bibinfo {volume} {94}},\ \bibinfo
  {pages} {010103(R)} (\bibinfo {year} {2016})}\BibitemShut {NoStop}%
\bibitem [{\citenamefont {\AA{}berg}(2018)}]{Aberg_2018}%
  \BibitemOpen
  \bibfield  {author} {\bibinfo {author} {\bibfnamefont {J.}~\bibnamefont
  {\AA{}berg}},\ }\bibfield  {title} {\emph {\bibinfo {title} {Fully quantum
  fluctuation theorems},\ }}\href {\doibase10.1103/PhysRevX.8.011019}
  {\bibfield  {journal} {\bibinfo  {journal} {Phys. Rev. X}\ }\textbf {\bibinfo
  {volume} {8}},\ \bibinfo {pages} {011019} (\bibinfo {year}
  {2018})}\BibitemShut {NoStop}%
\bibitem [{\citenamefont {Alhambra}\ \emph {et~al.}(2016)\citenamefont
  {Alhambra}, \citenamefont {Masanes}, \citenamefont {Oppenheim},\ and\
  \citenamefont {Perry}}]{Alhambra_2016}%
  \BibitemOpen
  \bibfield  {author} {\bibinfo {author} {\bibfnamefont {A.~M.}\ \bibnamefont
  {Alhambra}}, \bibinfo {author} {\bibfnamefont {L.}~\bibnamefont {Masanes}},
  \bibinfo {author} {\bibfnamefont {J.}~\bibnamefont {Oppenheim}}, \ and\
  \bibinfo {author} {\bibfnamefont {C.}~\bibnamefont {Perry}},\ }\bibfield
  {title} {\emph {\bibinfo {title} {Fluctuating work: From quantum
  thermodynamical identities to a second law equality},\ }}\href
  {\doibase10.1103/PhysRevX.6.041017} {\bibfield  {journal} {\bibinfo
  {journal} {Phys. Rev. X}\ }\textbf {\bibinfo {volume} {6}},\ \bibinfo {pages}
  {041017} (\bibinfo {year} {2016})}\BibitemShut {NoStop}%
\bibitem [{\citenamefont {Miller}\ and\ \citenamefont
  {Anders}(2017)}]{Miller_2017}%
  \BibitemOpen
  \bibfield  {author} {\bibinfo {author} {\bibfnamefont {H.~J.~D.}\
  \bibnamefont {Miller}}\ and\ \bibinfo {author} {\bibfnamefont
  {J.}~\bibnamefont {Anders}},\ }\bibfield  {title} {\emph {\bibinfo {title}
  {Time-reversal symmetric work distributions for closed quantum dynamics in
  the histories framework},\ }}\href {\doibase10.1088/1367-2630/aa703f}
  {\bibfield  {journal} {\bibinfo  {journal} {New J. Phys.}\ }\textbf {\bibinfo
  {volume} {19}},\ \bibinfo {pages} {062001} (\bibinfo {year}
  {2017})}\BibitemShut {NoStop}%
\bibitem [{\citenamefont {Xu}\ \emph {et~al.}(2018)\citenamefont {Xu},
  \citenamefont {Zou}, \citenamefont {Guo},\ and\ \citenamefont
  {Kong}}]{Xu_2018}%
  \BibitemOpen
  \bibfield  {author} {\bibinfo {author} {\bibfnamefont {B.-M.}\ \bibnamefont
  {Xu}}, \bibinfo {author} {\bibfnamefont {J.}~\bibnamefont {Zou}}, \bibinfo
  {author} {\bibfnamefont {L.-S.}\ \bibnamefont {Guo}}, \ and\ \bibinfo
  {author} {\bibfnamefont {X.-M.}\ \bibnamefont {Kong}},\ }\bibfield  {title}
  {\emph {\bibinfo {title} {Effects of quantum coherence on work statistics},\
  }}\href {\doibase10.1103/PhysRevA.97.052122} {\bibfield  {journal} {\bibinfo
  {journal} {Phys. Rev. A}\ }\textbf {\bibinfo {volume} {97}},\ \bibinfo
  {pages} {052122} (\bibinfo {year} {2018})}\BibitemShut {NoStop}%
\bibitem [{\citenamefont {Gherardini}\ \emph {et~al.}(2021)\citenamefont
  {Gherardini}, \citenamefont {Belenchia}, \citenamefont {Paternostro},\ and\
  \citenamefont {Trombettoni}}]{Gherardini_2021}%
  \BibitemOpen
  \bibfield  {author} {\bibinfo {author} {\bibfnamefont {S.}~\bibnamefont
  {Gherardini}}, \bibinfo {author} {\bibfnamefont {A.}~\bibnamefont
  {Belenchia}}, \bibinfo {author} {\bibfnamefont {M.}~\bibnamefont
  {Paternostro}}, \ and\ \bibinfo {author} {\bibfnamefont {A.}~\bibnamefont
  {Trombettoni}},\ }\bibfield  {title} {\emph {\bibinfo {title} {End-point
  measurement approach to assess quantum coherence in energy fluctuations},\
  }}\href {\doibase10.1103/PhysRevA.104.L050203} {\bibfield  {journal}
  {\bibinfo  {journal} {Phys. Rev. A}\ }\textbf {\bibinfo {volume} {104}},\
  \bibinfo {pages} {L050203} (\bibinfo {year} {2021})}\BibitemShut {NoStop}%
\bibitem [{\citenamefont {Beyer}\ \emph {et~al.}(2020)\citenamefont {Beyer},
  \citenamefont {Luoma},\ and\ \citenamefont {Strunz}}]{Beyer_2020}%
  \BibitemOpen
  \bibfield  {author} {\bibinfo {author} {\bibfnamefont {K.}~\bibnamefont
  {Beyer}}, \bibinfo {author} {\bibfnamefont {K.}~\bibnamefont {Luoma}}, \ and\
  \bibinfo {author} {\bibfnamefont {W.~T.}\ \bibnamefont {Strunz}},\ }\bibfield
   {title} {\emph {\bibinfo {title} {Work as an external quantum observable and
  an operational quantum work fluctuation theorem},\ }}\href
  {\doibase10.1103/PhysRevResearch.2.033508} {\bibfield  {journal} {\bibinfo
  {journal} {Phys. Rev. Research}\ }\textbf {\bibinfo {volume} {2}},\ \bibinfo
  {pages} {033508} (\bibinfo {year} {2020})}\BibitemShut {NoStop}%
\bibitem [{\citenamefont {Micadei}\ \emph {et~al.}()\citenamefont {Micadei},
  \citenamefont {Landi},\ and\ \citenamefont {Lutz}}]{Micadei_2021}%
  \BibitemOpen
  \bibfield  {author} {\bibinfo {author} {\bibfnamefont {K.}~\bibnamefont
  {Micadei}}, \bibinfo {author} {\bibfnamefont {G.~T.}\ \bibnamefont {Landi}},
  \ and\ \bibinfo {author} {\bibfnamefont {E.}~\bibnamefont {Lutz}},\
  }\bibfield  {title} {\emph {\bibinfo {title} {Extracting {B}ayesian networks
  from multiple copies of a quantum system},\ }}\href@noop {} {\ }\Eprint
  {https://doi.org/10.48550/arXiv.2103.14570}{arXiv:2103.14570}\BibitemShut
  {NoStop}%
\bibitem [{\citenamefont {Kerremans}\ \emph {et~al.}(2022)\citenamefont
  {Kerremans}, \citenamefont {Samuelsson},\ and\ \citenamefont
  {Potts}}]{Kerremans_2022}%
  \BibitemOpen
  \bibfield  {author} {\bibinfo {author} {\bibfnamefont {T.}~\bibnamefont
  {Kerremans}}, \bibinfo {author} {\bibfnamefont {P.}~\bibnamefont
  {Samuelsson}}, \ and\ \bibinfo {author} {\bibfnamefont {P.~P.}\ \bibnamefont
  {Potts}},\ }\bibfield  {title} {\emph {\bibinfo {title} {Probabilistically
  violating the first law of thermodynamics in a quantum heat engine},\ }}\href
  {\doibase10.21468/SciPostPhys.12.5.168} {\bibfield  {journal} {\bibinfo
  {journal} {SciPost Phys.}\ }\textbf {\bibinfo {volume} {12}},\ \bibinfo
  {pages} {168} (\bibinfo {year} {2022})}\BibitemShut {NoStop}%
\bibitem [{\citenamefont {Janovitch}\ and\ \citenamefont
  {Landi}(2022)}]{Janovitch_2022}%
  \BibitemOpen
  \bibfield  {author} {\bibinfo {author} {\bibfnamefont {M.}~\bibnamefont
  {Janovitch}}\ and\ \bibinfo {author} {\bibfnamefont {G.~T.}\ \bibnamefont
  {Landi}},\ }\bibfield  {title} {\emph {\bibinfo {title} {Quantum mean-square
  predictors and thermodynamics},\ }}\href
  {\doibase10.1103/PhysRevA.105.022217} {\bibfield  {journal} {\bibinfo
  {journal} {Phys. Rev. A}\ }\textbf {\bibinfo {volume} {105}},\ \bibinfo
  {pages} {022217} (\bibinfo {year} {2022})}\BibitemShut {NoStop}%
\bibitem [{\citenamefont {Beyer}\ \emph {et~al.}(2022)\citenamefont {Beyer},
  \citenamefont {Uola}, \citenamefont {Luoma},\ and\ \citenamefont
  {Strunz}}]{Beyer_2022}%
  \BibitemOpen
  \bibfield  {author} {\bibinfo {author} {\bibfnamefont {K.}~\bibnamefont
  {Beyer}}, \bibinfo {author} {\bibfnamefont {R.}~\bibnamefont {Uola}},
  \bibinfo {author} {\bibfnamefont {K.}~\bibnamefont {Luoma}}, \ and\ \bibinfo
  {author} {\bibfnamefont {W.~T.}\ \bibnamefont {Strunz}},\ }\bibfield  {title}
  {\emph {\bibinfo {title} {Joint measurability in nonequilibrium quantum
  thermodynamics},\ }}\href {\doibase10.1103/PhysRevE.106.L022101} {\bibfield
  {journal} {\bibinfo  {journal} {Phys. Rev. E}\ }\textbf {\bibinfo {volume}
  {106}},\ \bibinfo {pages} {L022101} (\bibinfo {year} {2022})}\BibitemShut
  {NoStop}%
\bibitem [{\citenamefont {Pei}\ \emph {et~al.}(2023)\citenamefont {Pei},
  \citenamefont {Chen},\ and\ \citenamefont {Quan}}]{Pei_2023}%
  \BibitemOpen
  \bibfield  {author} {\bibinfo {author} {\bibfnamefont {J.-H.}\ \bibnamefont
  {Pei}}, \bibinfo {author} {\bibfnamefont {J.-F.}\ \bibnamefont {Chen}}, \
  and\ \bibinfo {author} {\bibfnamefont {H.~T.}\ \bibnamefont {Quan}},\
  }\bibfield  {title} {\emph {\bibinfo {title} {Exploring quasiprobability
  approaches to quantum work in the presence of initial coherence: Advantages
  of the {Margenau--Hill} distribution},\ }}\href
  {\doibase10.1103/PhysRevE.108.054109} {\bibfield  {journal} {\bibinfo
  {journal} {Phys. Rev. E}\ }\textbf {\bibinfo {volume} {108}},\ \bibinfo
  {pages} {054109} (\bibinfo {year} {2023})}\BibitemShut {NoStop}%
\bibitem [{\citenamefont {Lostaglio}(2018)}]{Lostaglio_2018}%
  \BibitemOpen
  \bibfield  {author} {\bibinfo {author} {\bibfnamefont {M.}~\bibnamefont
  {Lostaglio}},\ }\bibfield  {title} {\emph {\bibinfo {title} {Quantum
  fluctuation theorems, contextuality, and work quasiprobabilities},\ }}\href
  {\doibase10.1103/PhysRevLett.120.040602} {\bibfield  {journal} {\bibinfo
  {journal} {Phys. Rev. Lett.}\ }\textbf {\bibinfo {volume} {120}},\ \bibinfo
  {pages} {040602} (\bibinfo {year} {2018})}\BibitemShut {NoStop}%
\bibitem [{\citenamefont {Seifert}(2012)}]{Seifert_2012}%
  \BibitemOpen
  \bibfield  {author} {\bibinfo {author} {\bibfnamefont {U.}~\bibnamefont
  {Seifert}},\ }\bibfield  {title} {\emph {\bibinfo {title} {Stochastic
  thermodynamics, fluctuation theorems and molecular machines},\ }}\href
  {\doibase10.1088/0034-4885/75/12/126001} {\bibfield  {journal} {\bibinfo
  {journal} {Rep. Prog. Phys.}\ }\textbf {\bibinfo {volume} {75}},\ \bibinfo
  {pages} {126001} (\bibinfo {year} {2012})}\BibitemShut {NoStop}%
\bibitem [{\citenamefont {Campisi}\ \emph {et~al.}(2011)\citenamefont
  {Campisi}, \citenamefont {H\"{a}nggi},\ and\ \citenamefont
  {Talkner}}]{Campisi_2011}%
  \BibitemOpen
  \bibfield  {author} {\bibinfo {author} {\bibfnamefont {M.}~\bibnamefont
  {Campisi}}, \bibinfo {author} {\bibfnamefont {P.}~\bibnamefont {H\"{a}nggi}},
  \ and\ \bibinfo {author} {\bibfnamefont {P.}~\bibnamefont {Talkner}},\
  }\bibfield  {title} {\emph {\bibinfo {title} {Colloquium: {Q}uantum
  fluctuation relations: {F}oundations and applications},\ }}\href
  {\doibase10.1103/RevModPhys.83.771} {\bibfield  {journal} {\bibinfo
  {journal} {Rev. Mod. Phys.}\ }\textbf {\bibinfo {volume} {83}},\ \bibinfo
  {pages} {771} (\bibinfo {year} {2011})}\BibitemShut {NoStop}%
\bibitem [{\citenamefont {Landau}\ and\ \citenamefont {Lifshitz}(1980)}]{ll5}%
  \BibitemOpen
  \bibfield  {author} {\bibinfo {author} {\bibfnamefont {L.~D.}\ \bibnamefont
  {Landau}}\ and\ \bibinfo {author} {\bibfnamefont {E.~M.}\ \bibnamefont
  {Lifshitz}},\ }\href@noop {} {\emph {\bibinfo {title} {Statistical Physics,
  Part I}}}\ (\bibinfo  {publisher} {Pergamon, New York},\ \bibinfo {year}
  {1980})\BibitemShut {NoStop}%
\bibitem [{\citenamefont {Lindblad}(1983)}]{Lindblad_book}%
  \BibitemOpen
  \bibfield  {author} {\bibinfo {author} {\bibfnamefont {G.}~\bibnamefont
  {Lindblad}},\ }\href@noop {} {\emph {\bibinfo {title} {Non-Equilibrium
  Entropy and Irreversibility}}}\ (\bibinfo  {publisher} {Reidel, Dordrecht},\
  \bibinfo {year} {1983})\BibitemShut {NoStop}%
\bibitem [{\citenamefont {Allahverdyan}\ \emph {et~al.}(2013)\citenamefont
  {Allahverdyan}, \citenamefont {Balian},\ and\ \citenamefont
  {Nieuwenhuizen}}]{Allahverdyan_2013}%
  \BibitemOpen
  \bibfield  {author} {\bibinfo {author} {\bibfnamefont {A.~E.}\ \bibnamefont
  {Allahverdyan}}, \bibinfo {author} {\bibfnamefont {R.}~\bibnamefont
  {Balian}}, \ and\ \bibinfo {author} {\bibfnamefont {T.~M.}\ \bibnamefont
  {Nieuwenhuizen}},\ }\bibfield  {title} {\emph {\bibinfo {title}
  {Understanding quantum measurement from the solution of dynamical models},\
  }}\href {\doibase10.1016/j.physrep.2012.11.001} {\bibfield  {journal}
  {\bibinfo  {journal} {Phys. Reps.}\ }\textbf {\bibinfo {volume} {525}},\
  \bibinfo {pages} {1} (\bibinfo {year} {2013})}\BibitemShut {NoStop}%
\bibitem [{\citenamefont {Masanes}\ \emph {et~al.}(2019)\citenamefont
  {Masanes}, \citenamefont {Galley},\ and\ \citenamefont
  {M\"{u}ller}}]{Masanes_2019}%
  \BibitemOpen
  \bibfield  {author} {\bibinfo {author} {\bibfnamefont {L.}~\bibnamefont
  {Masanes}}, \bibinfo {author} {\bibfnamefont {T.~D.}\ \bibnamefont {Galley}},
  \ and\ \bibinfo {author} {\bibfnamefont {M.~P.}\ \bibnamefont {M\"{u}ller}},\
  }\bibfield  {title} {\emph {\bibinfo {title} {The measurement postulates of
  quantum mechanics are operationally redundant},\ }}\href
  {\doibase10.1038/s41467-019-09348-x} {\bibfield  {journal} {\bibinfo
  {journal} {Nat. Commun.}\ }\textbf {\bibinfo {volume} {10}},\ \bibinfo
  {pages} {1361} (\bibinfo {year} {2019})}\BibitemShut {NoStop}%
\bibitem [{\citenamefont {Abdelkhalek}\ \emph {et~al.}()\citenamefont
  {Abdelkhalek}, \citenamefont {Nakata},\ and\ \citenamefont
  {Reeb}}]{Abdelkhalek_2016}%
  \BibitemOpen
  \bibfield  {author} {\bibinfo {author} {\bibfnamefont {K.}~\bibnamefont
  {Abdelkhalek}}, \bibinfo {author} {\bibfnamefont {Y.}~\bibnamefont {Nakata}},
  \ and\ \bibinfo {author} {\bibfnamefont {D.}~\bibnamefont {Reeb}},\
  }\bibfield  {title} {\emph {\bibinfo {title} {Fundamental energy cost for
  quantum measurement},\ }}\href@noop {} {\ }\Eprint
  {https://doi.org/10.48550/arXiv.1609.06981}{arXiv:1609.06981}\BibitemShut
  {NoStop}%
\bibitem [{\citenamefont {Nielsen}\ and\ \citenamefont
  {Chuang}(2010)}]{mikeike}%
  \BibitemOpen
  \bibfield  {author} {\bibinfo {author} {\bibfnamefont {M.~A.}\ \bibnamefont
  {Nielsen}}\ and\ \bibinfo {author} {\bibfnamefont {I.~L.}\ \bibnamefont
  {Chuang}},\ }\href@noop {} {\emph {\bibinfo {title} {Quantum computation and
  quantum information}}}\ (\bibinfo  {publisher} {Cambridge University Press,
  Cambridge, England},\ \bibinfo {year} {2010})\BibitemShut {NoStop}%
\bibitem [{\citenamefont {Jarzynski}\ \emph {et~al.}(2015)\citenamefont
  {Jarzynski}, \citenamefont {Quan},\ and\ \citenamefont
  {Rahav}}]{Jarzynski_2015}%
  \BibitemOpen
  \bibfield  {author} {\bibinfo {author} {\bibfnamefont {C.}~\bibnamefont
  {Jarzynski}}, \bibinfo {author} {\bibfnamefont {H.~T.}\ \bibnamefont {Quan}},
  \ and\ \bibinfo {author} {\bibfnamefont {S.}~\bibnamefont {Rahav}},\
  }\bibfield  {title} {\emph {\bibinfo {title} {Quantum-classical
  correspondence principle for work distributions},\ }}\href
  {\doibase10.1103/PhysRevX.5.031038} {\bibfield  {journal} {\bibinfo
  {journal} {Phys. Rev. X}\ }\textbf {\bibinfo {volume} {5}},\ \bibinfo {pages}
  {031038} (\bibinfo {year} {2015})}\BibitemShut {NoStop}%
\bibitem [{\citenamefont {Garc\'{i}a-Mata}\ \emph {et~al.}(2017)\citenamefont
  {Garc\'{i}a-Mata}, \citenamefont {Roncaglia},\ and\ \citenamefont
  {Wisniacki}}]{Garcia-Mata_2017}%
  \BibitemOpen
  \bibfield  {author} {\bibinfo {author} {\bibfnamefont {I.}~\bibnamefont
  {Garc\'{i}a-Mata}}, \bibinfo {author} {\bibfnamefont {A.~J.}\ \bibnamefont
  {Roncaglia}}, \ and\ \bibinfo {author} {\bibfnamefont {D.~A.}\ \bibnamefont
  {Wisniacki}},\ }\bibfield  {title} {\emph {\bibinfo {title} {Semiclassical
  approach to the work distribution},\ }}\href
  {\doibase10.1209/0295-5075/120/30002} {\bibfield  {journal} {\bibinfo
  {journal} {Europhys. Lett.}\ }\textbf {\bibinfo {volume} {120}},\ \bibinfo
  {pages} {30002} (\bibinfo {year} {2017})}\BibitemShut {NoStop}%
\bibitem [{\citenamefont {Funo}\ and\ \citenamefont {Quan}(2018)}]{Funo_2018}%
  \BibitemOpen
  \bibfield  {author} {\bibinfo {author} {\bibfnamefont {K.}~\bibnamefont
  {Funo}}\ and\ \bibinfo {author} {\bibfnamefont {H.~T.}\ \bibnamefont
  {Quan}},\ }\bibfield  {title} {\emph {\bibinfo {title} {Path integral
  approach to quantum thermodynamics},\ }}\href
  {\doibase10.1103/PhysRevLett.121.040602} {\bibfield  {journal} {\bibinfo
  {journal} {Phys. Rev. Lett.}\ }\textbf {\bibinfo {volume} {121}},\ \bibinfo
  {pages} {040602} (\bibinfo {year} {2018})}\BibitemShut {NoStop}%
\bibitem [{\citenamefont {Fei}\ \emph {et~al.}(2018)\citenamefont {Fei},
  \citenamefont {Quan},\ and\ \citenamefont {Liu}}]{Fei_2018}%
  \BibitemOpen
  \bibfield  {author} {\bibinfo {author} {\bibfnamefont {Z.}~\bibnamefont
  {Fei}}, \bibinfo {author} {\bibfnamefont {H.~T.}\ \bibnamefont {Quan}}, \
  and\ \bibinfo {author} {\bibfnamefont {F.}~\bibnamefont {Liu}},\ }\bibfield
  {title} {\emph {\bibinfo {title} {Quantum corrections of work statistics in
  closed quantum systems},\ }}\href {\doibase10.1103/PhysRevE.98.012132}
  {\bibfield  {journal} {\bibinfo  {journal} {Phys. Rev. E}\ }\textbf {\bibinfo
  {volume} {98}},\ \bibinfo {pages} {012132} (\bibinfo {year}
  {2018})}\BibitemShut {NoStop}%
\bibitem [{\citenamefont {Petz}(1994)}]{Petz_1994}%
  \BibitemOpen
  \bibfield  {author} {\bibinfo {author} {\bibfnamefont {D.}~\bibnamefont
  {Petz}},\ }\bibfield  {title} {\emph {\bibinfo {title} {A survey of certain
  trace inequalities},\ }}\href {\doibase10.4064/-30-1-287-298} {\bibfield
  {journal} {\bibinfo  {journal} {Banach Cent. Publ.}\ }\textbf {\bibinfo
  {volume} {30}},\ \bibinfo {pages} {287} (\bibinfo {year} {1994})}\BibitemShut
  {NoStop}%
\bibitem [{\citenamefont {Pan}\ \emph {et~al.}()\citenamefont {Pan},
  \citenamefont {Fei}, \citenamefont {Qiu}, \citenamefont {Zhang},\ and\
  \citenamefont {Quan}}]{Pan_2019}%
  \BibitemOpen
  \bibfield  {author} {\bibinfo {author} {\bibfnamefont {R.}~\bibnamefont
  {Pan}}, \bibinfo {author} {\bibfnamefont {Z.}~\bibnamefont {Fei}}, \bibinfo
  {author} {\bibfnamefont {T.}~\bibnamefont {Qiu}}, \bibinfo {author}
  {\bibfnamefont {J.-N.}\ \bibnamefont {Zhang}}, \ and\ \bibinfo {author}
  {\bibfnamefont {H.~T.}\ \bibnamefont {Quan}},\ }\bibfield  {title} {\emph
  {\bibinfo {title} {Quantum-classical correspondence of work distributions for
  initial states with quantum coherence},\ }}\href@noop {} {\ }\Eprint
  {https://doi.org/10.48550/arXiv.1904.05378}{arXiv:1904.05378}\BibitemShut
  {NoStop}%
\bibitem [{\citenamefont {Huber}\ \emph {et~al.}(2008)\citenamefont {Huber},
  \citenamefont {Schmidt-Kaler}, \citenamefont {Deffner},\ and\ \citenamefont
  {Lutz}}]{Huber_2008}%
  \BibitemOpen
  \bibfield  {author} {\bibinfo {author} {\bibfnamefont {G.}~\bibnamefont
  {Huber}}, \bibinfo {author} {\bibfnamefont {F.}~\bibnamefont
  {Schmidt-Kaler}}, \bibinfo {author} {\bibfnamefont {S.}~\bibnamefont
  {Deffner}}, \ and\ \bibinfo {author} {\bibfnamefont {E.}~\bibnamefont
  {Lutz}},\ }\bibfield  {title} {\emph {\bibinfo {title} {Employing trapped
  cold ions to verify the quantum {J}arzynski equality},\ }}\href
  {\doibase10.1103/PhysRevLett.101.070403} {\bibfield  {journal} {\bibinfo
  {journal} {Phys. Rev. Lett.}\ }\textbf {\bibinfo {volume} {101}},\ \bibinfo
  {pages} {070403} (\bibinfo {year} {2008})}\BibitemShut {NoStop}%
\bibitem [{\citenamefont {Batalh\~{a}o}\ \emph {et~al.}(2014)\citenamefont
  {Batalh\~{a}o}, \citenamefont {Souza}, \citenamefont {Mazzola}, \citenamefont
  {Auccaise}, \citenamefont {Sarthour}, \citenamefont {Oliveira}, \citenamefont
  {Goold}, \citenamefont {De~Chiara}, \citenamefont {Paternostro},\ and\
  \citenamefont {Serra}}]{Batalhao_2014}%
  \BibitemOpen
  \bibfield  {author} {\bibinfo {author} {\bibfnamefont {T.~B.}\ \bibnamefont
  {Batalh\~{a}o}}, \bibinfo {author} {\bibfnamefont {A.~M.}\ \bibnamefont
  {Souza}}, \bibinfo {author} {\bibfnamefont {L.}~\bibnamefont {Mazzola}},
  \bibinfo {author} {\bibfnamefont {R.}~\bibnamefont {Auccaise}}, \bibinfo
  {author} {\bibfnamefont {R.~S.}\ \bibnamefont {Sarthour}}, \bibinfo {author}
  {\bibfnamefont {I.~S.}\ \bibnamefont {Oliveira}}, \bibinfo {author}
  {\bibfnamefont {J.}~\bibnamefont {Goold}}, \bibinfo {author} {\bibfnamefont
  {G.}~\bibnamefont {De~Chiara}}, \bibinfo {author} {\bibfnamefont
  {M.}~\bibnamefont {Paternostro}}, \ and\ \bibinfo {author} {\bibfnamefont
  {R.~M.}\ \bibnamefont {Serra}},\ }\bibfield  {title} {\emph {\bibinfo {title}
  {Experimental reconstruction of work distribution and study of fluctuation
  relations in a closed quantum system},\ }}\href
  {\doibase10.1103/PhysRevLett.113.140601} {\bibfield  {journal} {\bibinfo
  {journal} {Phys. Rev. Lett.}\ }\textbf {\bibinfo {volume} {113}},\ \bibinfo
  {pages} {140601} (\bibinfo {year} {2014})}\BibitemShut {NoStop}%
\bibitem [{\citenamefont {An}\ \emph {et~al.}(2015)\citenamefont {An},
  \citenamefont {Zhang}, \citenamefont {Um}, \citenamefont {Lv}, \citenamefont
  {Lu}, \citenamefont {Zhang}, \citenamefont {Yin}, \citenamefont {Quan},\ and\
  \citenamefont {Kim}}]{An_2015}%
  \BibitemOpen
  \bibfield  {author} {\bibinfo {author} {\bibfnamefont {S.}~\bibnamefont
  {An}}, \bibinfo {author} {\bibfnamefont {J.-N.}\ \bibnamefont {Zhang}},
  \bibinfo {author} {\bibfnamefont {M.}~\bibnamefont {Um}}, \bibinfo {author}
  {\bibfnamefont {D.}~\bibnamefont {Lv}}, \bibinfo {author} {\bibfnamefont
  {Y.}~\bibnamefont {Lu}}, \bibinfo {author} {\bibfnamefont {J.}~\bibnamefont
  {Zhang}}, \bibinfo {author} {\bibfnamefont {Z.-Q.}\ \bibnamefont {Yin}},
  \bibinfo {author} {\bibfnamefont {H.~T.}\ \bibnamefont {Quan}}, \ and\
  \bibinfo {author} {\bibfnamefont {K.}~\bibnamefont {Kim}},\ }\bibfield
  {title} {\emph {\bibinfo {title} {Experimental test of the quantum
  {J}arzynski equality with a trapped-ion system},\ }}\href
  {\doibase10.1038/nphys3197} {\bibfield  {journal} {\bibinfo  {journal} {Nat.
  Phys.}\ }\textbf {\bibinfo {volume} {11}},\ \bibinfo {pages} {193} (\bibinfo
  {year} {2015})}\BibitemShut {NoStop}%
\bibitem [{\citenamefont {Allahverdyan}\ \emph {et~al.}(2004)\citenamefont
  {Allahverdyan}, \citenamefont {Balian},\ and\ \citenamefont
  {Nieuwenhuizen}}]{Allahverdyan_2004}%
  \BibitemOpen
  \bibfield  {author} {\bibinfo {author} {\bibfnamefont {A.~E.}\ \bibnamefont
  {Allahverdyan}}, \bibinfo {author} {\bibfnamefont {R.}~\bibnamefont
  {Balian}}, \ and\ \bibinfo {author} {\bibfnamefont {T.~M.}\ \bibnamefont
  {Nieuwenhuizen}},\ }\bibfield  {title} {\emph {\bibinfo {title} {Maximal work
  extraction from finite quantum systems},\ }}\href
  {\doibase10.1209/epl/i2004-10101-2} {\bibfield  {journal} {\bibinfo
  {journal} {Europhys. Lett.}\ }\textbf {\bibinfo {volume} {67}},\ \bibinfo
  {pages} {565} (\bibinfo {year} {2004})}\BibitemShut {NoStop}%
\bibitem [{\citenamefont {\v{S}afr\'{a}nek}\ \emph {et~al.}(2023)\citenamefont
  {\v{S}afr\'{a}nek}, \citenamefont {Rosa},\ and\ \citenamefont
  {Binder}}]{Safranek_2023}%
  \BibitemOpen
  \bibfield  {author} {\bibinfo {author} {\bibfnamefont {D.}~\bibnamefont
  {\v{S}afr\'{a}nek}}, \bibinfo {author} {\bibfnamefont {D.}~\bibnamefont
  {Rosa}}, \ and\ \bibinfo {author} {\bibfnamefont {F.~C.}\ \bibnamefont
  {Binder}},\ }\bibfield  {title} {\emph {\bibinfo {title} {Work extraction
  from unknown quantum sources},\ }}\href
  {\doibase10.1103/PhysRevLett.130.210401} {\bibfield  {journal} {\bibinfo
  {journal} {Phys. Rev. Lett.}\ }\textbf {\bibinfo {volume} {130}},\ \bibinfo
  {pages} {210401} (\bibinfo {year} {2023})}\BibitemShut {NoStop}%
\bibitem [{\citenamefont {Wiseman}(1995)}]{Wiseman_1995}%
  \BibitemOpen
  \bibfield  {author} {\bibinfo {author} {\bibfnamefont {H.~M.}\ \bibnamefont
  {Wiseman}},\ }\bibfield  {title} {\emph {\bibinfo {title} {Adaptive phase
  measurements of optical modes: Going beyond the marginal {$Q$}
  distribution},\ }}\href {\doibase10.1103/PhysRevLett.75.4587} {\bibfield
  {journal} {\bibinfo  {journal} {Phys. Rev. Lett.}\ }\textbf {\bibinfo
  {volume} {75}},\ \bibinfo {pages} {4587} (\bibinfo {year}
  {1995})}\BibitemShut {NoStop}%
\bibitem [{\citenamefont {Armen}\ \emph {et~al.}(2002)\citenamefont {Armen},
  \citenamefont {Au}, \citenamefont {Stockton}, \citenamefont {Doherty},\ and\
  \citenamefont {Mabuchi}}]{Armen_2002}%
  \BibitemOpen
  \bibfield  {author} {\bibinfo {author} {\bibfnamefont {M.~A.}\ \bibnamefont
  {Armen}}, \bibinfo {author} {\bibfnamefont {J.~K.}\ \bibnamefont {Au}},
  \bibinfo {author} {\bibfnamefont {J.~K.}\ \bibnamefont {Stockton}}, \bibinfo
  {author} {\bibfnamefont {A.~C.}\ \bibnamefont {Doherty}}, \ and\ \bibinfo
  {author} {\bibfnamefont {H.}~\bibnamefont {Mabuchi}},\ }\bibfield  {title}
  {\emph {\bibinfo {title} {Adaptive homodyne measurement of optical phase},\
  }}\href {\doibase10.1103/PhysRevLett.89.133602} {\bibfield  {journal}
  {\bibinfo  {journal} {Phys. Rev. Lett.}\ }\textbf {\bibinfo {volume} {89}},\
  \bibinfo {pages} {133602} (\bibinfo {year} {2002})}\BibitemShut {NoStop}%
\bibitem [{\citenamefont {O’Brien}\ \emph {et~al.}(2009)\citenamefont
  {O’Brien}, \citenamefont {Furusawa},\ and\ \citenamefont
  {Vu\v{c}kovi\'{c}}}]{OBrien_2009}%
  \BibitemOpen
  \bibfield  {author} {\bibinfo {author} {\bibfnamefont {J.~L.}\ \bibnamefont
  {O’Brien}}, \bibinfo {author} {\bibfnamefont {A.}~\bibnamefont {Furusawa}},
  \ and\ \bibinfo {author} {\bibfnamefont {J.}~\bibnamefont
  {Vu\v{c}kovi\'{c}}},\ }\bibfield  {title} {\emph {\bibinfo {title} {Photonic
  quantum technologies},\ }}\href {\doibase10.1038/nphoton.2009.229} {\bibfield
   {journal} {\bibinfo  {journal} {Nature Photonics}\ }\textbf {\bibinfo
  {volume} {3}},\ \bibinfo {pages} {687} (\bibinfo {year} {2009})}\BibitemShut
  {NoStop}%
\bibitem [{\citenamefont {Berni}\ \emph {et~al.}(2015)\citenamefont {Berni},
  \citenamefont {Gehring}, \citenamefont {Nielsen}, \citenamefont {Händchen},
  \citenamefont {Paris},\ and\ \citenamefont {Andersen}}]{Berni_2015}%
  \BibitemOpen
  \bibfield  {author} {\bibinfo {author} {\bibfnamefont {A.~A.}\ \bibnamefont
  {Berni}}, \bibinfo {author} {\bibfnamefont {T.}~\bibnamefont {Gehring}},
  \bibinfo {author} {\bibfnamefont {B.~M.}\ \bibnamefont {Nielsen}}, \bibinfo
  {author} {\bibfnamefont {V.}~\bibnamefont {Händchen}}, \bibinfo {author}
  {\bibfnamefont {M.~G.~A.}\ \bibnamefont {Paris}}, \ and\ \bibinfo {author}
  {\bibfnamefont {U.~L.}\ \bibnamefont {Andersen}},\ }\bibfield  {title} {\emph
  {\bibinfo {title} {Ab initio quantum-enhanced optical phase estimation using
  real-time feedback control},\ }}\href {\doibase10.1038/nphoton.2015.139}
  {\bibfield  {journal} {\bibinfo  {journal} {Nature Photonics}\ }\textbf
  {\bibinfo {volume} {9}},\ \bibinfo {pages} {577} (\bibinfo {year}
  {2015})}\BibitemShut {NoStop}%
\bibitem [{\citenamefont {Wu}\ \emph {et~al.}(2019)\citenamefont {Wu},
  \citenamefont {Yuan}, \citenamefont {Xiang}, \citenamefont {Li},
  \citenamefont {Guo},\ and\ \citenamefont {Perarnau-Llobet}}]{Wu_2019}%
  \BibitemOpen
  \bibfield  {author} {\bibinfo {author} {\bibfnamefont {K.-D.}\ \bibnamefont
  {Wu}}, \bibinfo {author} {\bibfnamefont {Y.}~\bibnamefont {Yuan}}, \bibinfo
  {author} {\bibfnamefont {G.-Y.}\ \bibnamefont {Xiang}}, \bibinfo {author}
  {\bibfnamefont {C.-F.}\ \bibnamefont {Li}}, \bibinfo {author} {\bibfnamefont
  {G.-C.}\ \bibnamefont {Guo}}, \ and\ \bibinfo {author} {\bibfnamefont
  {M.}~\bibnamefont {Perarnau-Llobet}},\ }\bibfield  {title} {\emph {\bibinfo
  {title} {Experimentally reducing the quantum measurement back action in work
  distributions by a collective measurement},\ }}\href
  {\doibase10.1126/sciadv.aav4944} {\bibfield  {journal} {\bibinfo  {journal}
  {Sci. Adv.}\ }\textbf {\bibinfo {volume} {5}},\ \bibinfo {pages} {4944}
  (\bibinfo {year} {2019})}\BibitemShut {NoStop}%
\bibitem [{\citenamefont {Wu}\ \emph {et~al.}(2020)\citenamefont {Wu},
  \citenamefont {B\"{a}umer}, \citenamefont {Tang}, \citenamefont
  {Hovhannisyan}, \citenamefont {Perarnau-Llobet}, \citenamefont {Xiang},
  \citenamefont {Li},\ and\ \citenamefont {Guo}}]{Wu_2020}%
  \BibitemOpen
  \bibfield  {author} {\bibinfo {author} {\bibfnamefont {K.-D.}\ \bibnamefont
  {Wu}}, \bibinfo {author} {\bibfnamefont {E.}~\bibnamefont {B\"{a}umer}},
  \bibinfo {author} {\bibfnamefont {J.-F.}\ \bibnamefont {Tang}}, \bibinfo
  {author} {\bibfnamefont {K.~V.}\ \bibnamefont {Hovhannisyan}}, \bibinfo
  {author} {\bibfnamefont {M.}~\bibnamefont {Perarnau-Llobet}}, \bibinfo
  {author} {\bibfnamefont {G.-Y.}\ \bibnamefont {Xiang}}, \bibinfo {author}
  {\bibfnamefont {C.-F.}\ \bibnamefont {Li}}, \ and\ \bibinfo {author}
  {\bibfnamefont {G.-C.}\ \bibnamefont {Guo}},\ }\bibfield  {title} {\emph
  {\bibinfo {title} {Minimizing backaction through entangled measurements},\
  }}\href {\doibase10.1103/PhysRevLett.125.210401} {\bibfield  {journal}
  {\bibinfo  {journal} {Phys. Rev. Lett.}\ }\textbf {\bibinfo {volume} {125}},\
  \bibinfo {pages} {210401} (\bibinfo {year} {2020})}\BibitemShut {NoStop}%
\bibitem [{\citenamefont {Horn}\ and\ \citenamefont {Johnson}(2013)}]{hojo}%
  \BibitemOpen
  \bibfield  {author} {\bibinfo {author} {\bibfnamefont {R.~A.}\ \bibnamefont
  {Horn}}\ and\ \bibinfo {author} {\bibfnamefont {C.~R.}\ \bibnamefont
  {Johnson}},\ }\href@noop {} {\emph {\bibinfo {title} {Matrix analysis}}},\
  \bibinfo {edition} {2nd}\ ed.\ (\bibinfo  {publisher} {Cambridge University
  Press, New York},\ \bibinfo {year} {2013})\BibitemShut {NoStop}%
\bibitem [{\citenamefont {Levy}\ and\ \citenamefont
  {Lostaglio}(2020)}]{Levy_2020}%
  \BibitemOpen
  \bibfield  {author} {\bibinfo {author} {\bibfnamefont {A.}~\bibnamefont
  {Levy}}\ and\ \bibinfo {author} {\bibfnamefont {M.}~\bibnamefont
  {Lostaglio}},\ }\bibfield  {title} {\emph {\bibinfo {title} {Quasiprobability
  distribution for heat fluctuations in the quantum regime},\ }}\href
  {\doibase10.1103/PRXQuantum.1.010309} {\bibfield  {journal} {\bibinfo
  {journal} {PRX Quantum}\ }\textbf {\bibinfo {volume} {1}},\ \bibinfo {pages}
  {010309} (\bibinfo {year} {2020})}\BibitemShut {NoStop}%
\bibitem [{\citenamefont {Mohammady}\ \emph {et~al.}(2020)\citenamefont
  {Mohammady}, \citenamefont {Auff\`{e}ves},\ and\ \citenamefont
  {Anders}}]{Mohammady_2020}%
  \BibitemOpen
  \bibfield  {author} {\bibinfo {author} {\bibfnamefont {M.~H.}\ \bibnamefont
  {Mohammady}}, \bibinfo {author} {\bibfnamefont {A.}~\bibnamefont
  {Auff\`{e}ves}}, \ and\ \bibinfo {author} {\bibfnamefont {J.}~\bibnamefont
  {Anders}},\ }\bibfield  {title} {\emph {\bibinfo {title} {Energetic
  footprints of irreversibility in the quantum regime},\ }}\href
  {\doibase10.1038/s42005-020-0356-9} {\bibfield  {journal} {\bibinfo
  {journal} {Commun. Phys.}\ }\textbf {\bibinfo {volume} {3}},\ \bibinfo
  {pages} {89} (\bibinfo {year} {2020})}\BibitemShut {NoStop}%
\bibitem [{\citenamefont {Nazarov}\ and\ \citenamefont
  {Kindermann}(2003)}]{Nazarov_2003}%
  \BibitemOpen
  \bibfield  {author} {\bibinfo {author} {\bibfnamefont {Y.~V.}\ \bibnamefont
  {Nazarov}}\ and\ \bibinfo {author} {\bibfnamefont {M.}~\bibnamefont
  {Kindermann}},\ }\bibfield  {title} {\emph {\bibinfo {title} {Full counting
  statistics of a general quantum mechanical variable},\ }}\href
  {\doibase10.1140/epjb/e2003-00293-1} {\bibfield  {journal} {\bibinfo
  {journal} {Eur. Phys. J. B}\ }\textbf {\bibinfo {volume} {35}},\ \bibinfo
  {pages} {413} (\bibinfo {year} {2003})}\BibitemShut {NoStop}%
\bibitem [{\citenamefont {Hofer}(2017)}]{Hofer_2017}%
  \BibitemOpen
  \bibfield  {author} {\bibinfo {author} {\bibfnamefont {P.~P.}\ \bibnamefont
  {Hofer}},\ }\bibfield  {title} {\emph {\bibinfo {title} {Quasi-probability
  distributions for observables in dynamic systems},\ }}\href
  {\doibase10.22331/q-2017-10-12-32} {\bibfield  {journal} {\bibinfo  {journal}
  {{Quantum}}\ }\textbf {\bibinfo {volume} {1}},\ \bibinfo {pages} {32}
  (\bibinfo {year} {2017})}\BibitemShut {NoStop}%
\bibitem [{\citenamefont {Hovhannisyan}\ and\ \citenamefont
  {Imparato}(2019)}]{Hovhannisyan_2019}%
  \BibitemOpen
  \bibfield  {author} {\bibinfo {author} {\bibfnamefont {K.~V.}\ \bibnamefont
  {Hovhannisyan}}\ and\ \bibinfo {author} {\bibfnamefont {A.}~\bibnamefont
  {Imparato}},\ }\bibfield  {title} {\emph {\bibinfo {title} {Quantum current
  in dissipative systems},\ }}\href {\doibase10.1088/1367-2630/ab1731}
  {\bibfield  {journal} {\bibinfo  {journal} {New J. Phys.}\ }\textbf {\bibinfo
  {volume} {21}},\ \bibinfo {pages} {052001} (\bibinfo {year}
  {2019})}\BibitemShut {NoStop}%
\bibitem [{\citenamefont {Elouard}\ \emph {et~al.}(2017)\citenamefont
  {Elouard}, \citenamefont {Herrera-Mart\'{i}}, \citenamefont {Clusel},\ and\
  \citenamefont {Auff\`{e}ves}}]{Elouard_2017}%
  \BibitemOpen
  \bibfield  {author} {\bibinfo {author} {\bibfnamefont {C.}~\bibnamefont
  {Elouard}}, \bibinfo {author} {\bibfnamefont {D.~A.}\ \bibnamefont
  {Herrera-Mart\'{i}}}, \bibinfo {author} {\bibfnamefont {M.}~\bibnamefont
  {Clusel}}, \ and\ \bibinfo {author} {\bibfnamefont {A.}~\bibnamefont
  {Auff\`{e}ves}},\ }\bibfield  {title} {\emph {\bibinfo {title} {The role of
  quantum measurement in stochastic thermodynamics},\ }}\href
  {\doibase10.1038/s41534-017-0008-4} {\bibfield  {journal} {\bibinfo
  {journal} {npj Quantum Inf.}\ }\textbf {\bibinfo {volume} {3}},\ \bibinfo
  {pages} {9} (\bibinfo {year} {2017})}\BibitemShut {NoStop}%
\bibitem [{\citenamefont {Manzano}\ \emph {et~al.}(2018)\citenamefont
  {Manzano}, \citenamefont {Horowitz},\ and\ \citenamefont
  {Parrondo}}]{Manzano_2018}%
  \BibitemOpen
  \bibfield  {author} {\bibinfo {author} {\bibfnamefont {G.}~\bibnamefont
  {Manzano}}, \bibinfo {author} {\bibfnamefont {J.~M.}\ \bibnamefont
  {Horowitz}}, \ and\ \bibinfo {author} {\bibfnamefont {J.~M.~R.}\ \bibnamefont
  {Parrondo}},\ }\bibfield  {title} {\emph {\bibinfo {title} {Quantum
  fluctuation theorems for arbitrary environments: Adiabatic and nonadiabatic
  entropy production},\ }}\href {\doibase10.1103/PhysRevX.8.031037} {\bibfield
  {journal} {\bibinfo  {journal} {Phys. Rev. X}\ }\textbf {\bibinfo {volume}
  {8}},\ \bibinfo {pages} {031037} (\bibinfo {year} {2018})}\BibitemShut
  {NoStop}%
\bibitem [{\citenamefont {De~Chiara}\ and\ \citenamefont
  {Imparato}(2022)}]{Chiara_2022}%
  \BibitemOpen
  \bibfield  {author} {\bibinfo {author} {\bibfnamefont {G.}~\bibnamefont
  {De~Chiara}}\ and\ \bibinfo {author} {\bibfnamefont {A.}~\bibnamefont
  {Imparato}},\ }\bibfield  {title} {\emph {\bibinfo {title} {Quantum
  fluctuation theorem for dissipative processes},\ }}\href
  {\doibase10.1103/PhysRevResearch.4.023230} {\bibfield  {journal} {\bibinfo
  {journal} {Phys. Rev. Res.}\ }\textbf {\bibinfo {volume} {4}},\ \bibinfo
  {pages} {023230} (\bibinfo {year} {2022})}\BibitemShut {NoStop}%
\bibitem [{\citenamefont {Sagawa}\ and\ \citenamefont
  {Ueda}(2009)}]{Sagawa_2009}%
  \BibitemOpen
  \bibfield  {author} {\bibinfo {author} {\bibfnamefont {T.}~\bibnamefont
  {Sagawa}}\ and\ \bibinfo {author} {\bibfnamefont {M.}~\bibnamefont {Ueda}},\
  }\bibfield  {title} {\emph {\bibinfo {title} {Minimal energy cost for
  thermodynamic information processing: Measurement and information erasure},\
  }}\href {\doibase10.1103/PhysRevLett.102.250602} {\bibfield  {journal}
  {\bibinfo  {journal} {Phys. Rev. Lett.}\ }\textbf {\bibinfo {volume} {102}},\
  \bibinfo {pages} {250602} (\bibinfo {year} {2009})}\BibitemShut {NoStop}%
\bibitem [{\citenamefont {Mohammady}\ and\ \citenamefont
  {Romito}(2019)}]{Mohammady_2019}%
  \BibitemOpen
  \bibfield  {author} {\bibinfo {author} {\bibfnamefont {M.~H.}\ \bibnamefont
  {Mohammady}}\ and\ \bibinfo {author} {\bibfnamefont {A.}~\bibnamefont
  {Romito}},\ }\bibfield  {title} {\emph {\bibinfo {title} {Conditional work
  statistics of quantum measurements},\ }}\href
  {\doibase10.22331/q-2019-08-19-175} {\bibfield  {journal} {\bibinfo
  {journal} {Quantum}\ }\textbf {\bibinfo {volume} {3}},\ \bibinfo {pages}
  {175} (\bibinfo {year} {2019})}\BibitemShut {NoStop}%
\bibitem [{\citenamefont {Guryanova}\ \emph {et~al.}(2020)\citenamefont
  {Guryanova}, \citenamefont {Friis},\ and\ \citenamefont
  {Huber}}]{Guryanova_2020}%
  \BibitemOpen
  \bibfield  {author} {\bibinfo {author} {\bibfnamefont {Y.}~\bibnamefont
  {Guryanova}}, \bibinfo {author} {\bibfnamefont {N.}~\bibnamefont {Friis}}, \
  and\ \bibinfo {author} {\bibfnamefont {M.}~\bibnamefont {Huber}},\ }\bibfield
   {title} {\emph {\bibinfo {title} {Ideal projective measurements have
  infinite resource costs},\ }}\href {\doibase10.22331/q-2020-01-13-222}
  {\bibfield  {journal} {\bibinfo  {journal} {{Quantum}}\ }\textbf {\bibinfo
  {volume} {4}},\ \bibinfo {pages} {222} (\bibinfo {year} {2020})}\BibitemShut
  {NoStop}%
\bibitem [{\citenamefont {Mohammady}(2023)}]{Mohammady_2022}%
  \BibitemOpen
  \bibfield  {author} {\bibinfo {author} {\bibfnamefont {M.~H.}\ \bibnamefont
  {Mohammady}},\ }\bibfield  {title} {\emph {\bibinfo {title}
  {Thermodynamically free quantum measurements},\ }}\href
  {\doibase10.1088/1751-8121/acad4a} {\bibfield  {journal} {\bibinfo  {journal}
  {J. Phys. A}\ }\textbf {\bibinfo {volume} {55}},\ \bibinfo {pages} {505304}
  (\bibinfo {year} {2023})}\BibitemShut {NoStop}%
\bibitem [{\citenamefont {Ito}\ \emph {et~al.}(2019)\citenamefont {Ito},
  \citenamefont {Talkner}, \citenamefont {Venkatesh},\ and\ \citenamefont
  {Watanabe}}]{Ito_2019}%
  \BibitemOpen
  \bibfield  {author} {\bibinfo {author} {\bibfnamefont {K.}~\bibnamefont
  {Ito}}, \bibinfo {author} {\bibfnamefont {P.}~\bibnamefont {Talkner}},
  \bibinfo {author} {\bibfnamefont {B.~P.}\ \bibnamefont {Venkatesh}}, \ and\
  \bibinfo {author} {\bibfnamefont {G.}~\bibnamefont {Watanabe}},\ }\bibfield
  {title} {\emph {\bibinfo {title} {Generalized energy measurements and quantum
  work compatible with fluctuation theorems},\ }}\href
  {\doibase10.1103/PhysRevA.99.032117} {\bibfield  {journal} {\bibinfo
  {journal} {Phys. Rev. A}\ }\textbf {\bibinfo {volume} {99}},\ \bibinfo
  {pages} {032117} (\bibinfo {year} {2019})}\BibitemShut {NoStop}%
\bibitem [{\citenamefont {Baker}(1905)}]{Baker_1905}%
  \BibitemOpen
  \bibfield  {author} {\bibinfo {author} {\bibfnamefont {H.~F.}\ \bibnamefont
  {Baker}},\ }\bibfield  {title} {\emph {\bibinfo {title} {Alternants and
  continuous groups},\ }}\href {\doibase10.1112/plms/s2-3.1.24} {\bibfield
  {journal} {\bibinfo  {journal} {Proc. London Math. Soc.}\ }\textbf {\bibinfo
  {volume} {s2-3}},\ \bibinfo {pages} {24} (\bibinfo {year}
  {1905})}\BibitemShut {NoStop}%
\bibitem [{\citenamefont {Lance}(1995)}]{Lance_1995}%
  \BibitemOpen
  \bibfield  {author} {\bibinfo {author} {\bibfnamefont {E.~C.}\ \bibnamefont
  {Lance}},\ }\href@noop {} {\emph {\bibinfo {title} {Hilbert {C}*-{M}odules:
  {A} {T}oolkit for {O}perator {A}lgebraists}}}\ (\bibinfo  {publisher}
  {Cambridge University Press, Cambridge},\ \bibinfo {year} {1995})\BibitemShut
  {NoStop}%
\bibitem [{\citenamefont {Jameson}(1996)}]{Jameson_1996}%
  \BibitemOpen
  \bibfield  {author} {\bibinfo {author} {\bibfnamefont {G.}~\bibnamefont
  {Jameson}},\ }\bibfield  {title} {\emph {\bibinfo {title} {Khinchin’s
  inequality for operators},\ }}\href {\doibase10.1017/S001708950003175X}
  {\bibfield  {journal} {\bibinfo  {journal} {Glasgow Math. J.}\ }\textbf
  {\bibinfo {volume} {38}},\ \bibinfo {pages} {327} (\bibinfo {year}
  {1996})}\BibitemShut {NoStop}%
\bibitem [{\citenamefont {Pusz}\ and\ \citenamefont
  {Woronowicz}(1978)}]{Pusz_1978}%
  \BibitemOpen
  \bibfield  {author} {\bibinfo {author} {\bibfnamefont {W.}~\bibnamefont
  {Pusz}}\ and\ \bibinfo {author} {\bibfnamefont {S.~L.}\ \bibnamefont
  {Woronowicz}},\ }\bibfield  {title} {\emph {\bibinfo {title} {Passive states
  and {KMS} states for general quantum systems},\ }}\href
  {\doibase10.1007/BF01614224} {\bibfield  {journal} {\bibinfo  {journal}
  {Commun. Math. Phys.}\ }\textbf {\bibinfo {volume} {58}},\ \bibinfo {pages}
  {273} (\bibinfo {year} {1978})}\BibitemShut {NoStop}%
\bibitem [{\citenamefont {Lenard}(1978)}]{Lenard_1978}%
  \BibitemOpen
  \bibfield  {author} {\bibinfo {author} {\bibfnamefont {A.}~\bibnamefont
  {Lenard}},\ }\bibfield  {title} {\emph {\bibinfo {title} {Thermodynamical
  proof of the {G}ibbs formula for elementary quantum systems},\ }}\href
  {\doibase10.1007/BF01011769} {\bibfield  {journal} {\bibinfo  {journal} {J.
  Stat. Phys.}\ }\textbf {\bibinfo {volume} {19}},\ \bibinfo {pages} {575}
  (\bibinfo {year} {1978})}\BibitemShut {NoStop}%
\bibitem [{\citenamefont {Lenard}(1971)}]{Lenard_1971}%
  \BibitemOpen
  \bibfield  {author} {\bibinfo {author} {\bibfnamefont {A.}~\bibnamefont
  {Lenard}},\ }\bibfield  {title} {\emph {\bibinfo {title} {Generalization of
  the {G}olden-{T}hompson inequality $\mathrm{Tr}(e^{A} e^{B}) \geqq
  \mathrm{Tr} e^{A+B}$},\ }}\href {\doibase10.1512/iumj.1972.21.21036}
  {\bibfield  {journal} {\bibinfo  {journal} {Indiana Univ. Math. J.}\ }\textbf
  {\bibinfo {volume} {21}},\ \bibinfo {pages} {457} (\bibinfo {year}
  {1971})}\BibitemShut {NoStop}%
\bibitem [{\citenamefont {Thompson}(1971)}]{Thompson_1971}%
  \BibitemOpen
  \bibfield  {author} {\bibinfo {author} {\bibfnamefont {C.~J.}\ \bibnamefont
  {Thompson}},\ }\bibfield  {title} {\emph {\bibinfo {title} {Inequalities and
  partial orders on matrix spaces},\ }}\href
  {\doibase10.1512/iumj.1972.21.21037} {\bibfield  {journal} {\bibinfo
  {journal} {Indiana Univ. Math. J.}\ }\textbf {\bibinfo {volume} {21}},\
  \bibinfo {pages} {469} (\bibinfo {year} {1971})}\BibitemShut {NoStop}%
\bibitem [{\citenamefont {Adesso}\ and\ \citenamefont
  {Illuminati}(2007)}]{Adesso_2007}%
  \BibitemOpen
  \bibfield  {author} {\bibinfo {author} {\bibfnamefont {G.}~\bibnamefont
  {Adesso}}\ and\ \bibinfo {author} {\bibfnamefont {F.}~\bibnamefont
  {Illuminati}},\ }\bibfield  {title} {\emph {\bibinfo {title} {Entanglement in
  continuous-variable systems: recent advances and current perspectives},\
  }}\href {\doibase10.1088/1751-8113/40/28/S01} {\bibfield  {journal} {\bibinfo
   {journal} {J. Phys. A}\ }\textbf {\bibinfo {volume} {40}},\ \bibinfo {pages}
  {7821} (\bibinfo {year} {2007})}\BibitemShut {NoStop}%
\bibitem [{\citenamefont {Williamson}(1936)}]{Williamson_1936}%
  \BibitemOpen
  \bibfield  {author} {\bibinfo {author} {\bibfnamefont {J.}~\bibnamefont
  {Williamson}},\ }\bibfield  {title} {\emph {\bibinfo {title} {On the
  algebraic problem concerning the normal forms of linear dynamical systems},\
  }}\href {\doibase10.2307/2371062} {\bibfield  {journal} {\bibinfo  {journal}
  {Am. J. Math.}\ }\textbf {\bibinfo {volume} {58}},\ \bibinfo {pages} {141}
  (\bibinfo {year} {1936})}\BibitemShut {NoStop}%
\bibitem [{\citenamefont {Scutaru}(1998)}]{Scutaru_1998}%
  \BibitemOpen
  \bibfield  {author} {\bibinfo {author} {\bibfnamefont {H.}~\bibnamefont
  {Scutaru}},\ }\bibfield  {title} {\emph {\bibinfo {title} {Fidelity for
  displaced squeezed thermal states and the oscillator semigroup},\ }}\href
  {\doibase10.1088/0305-4470/31/15/025} {\bibfield  {journal} {\bibinfo
  {journal} {J. Phys. A}\ }\textbf {\bibinfo {volume} {31}},\ \bibinfo {pages}
  {3659} (\bibinfo {year} {1998})}\BibitemShut {NoStop}%
\end{thebibliography}%

\clearpage

\renewcommand{\appendixname}{}

\onecolumngrid
\begin{center}
\textbf{\large Supplementary Note}
\end{center}

\vskip 1 cm

\twocolumngrid


\setcounter{section}{0}
\setcounter{equation}{0}
\setcounter{figure}{0}
\setcounter{table}{0}
\makeatletter
\renewcommand{\theequation}{S\arabic{equation}}
\renewcommand{\thefigure}{S\arabic{figure}}
\renewcommand{\thesection}{S}

\section{Classical limit of the operator of work}
\label{suppl:HOW_clasl}

Here, we will explore the question of whether the operator of work starts obeying the Jarzynski equality in the classical limit of a quantum continuous-variable system. We will study the simplest such example---the driven harmonic oscillator.

From the Baker--Cambell--Hausdorff formula,
\beaa \label{eq:bch}
\av{e^{-\beta \Omega}}_{\tau_\beta} =& \frac{1}{Z_\beta [H]} \tr \big( e^{ - \beta H} e^{- \beta \Omega} \big)
\\
=& \frac{1}{Z_\beta [H]} \tr e^{- \beta (H + \Omega) + \frac{1}{2} \beta^2 [H, \Omega] + \cdots}
\\
=& \frac{1}{Z_\beta [H]} \tr e^{- \beta U^\dagger H' U + \frac{1}{2} \beta^2 [H, U^\dagger H' U] + \cdots},
\eeaa
one could naively expect that, since the commutator terms go to zero as $\hbar \to 0$, one would have that $\av{e^{-\beta \Omega}}_{\tau_\beta}$ will tend to
\begin{align} \nonumber
\frac{1}{Z_\beta [H]} \tr e^{- \beta U^\dagger H' U} = e^{-\beta (F_\beta [H'] - F_\beta [H])}.
\end{align}
However, due to the presence of infinite-norm operators, the $\hbar \to 0$ limit may not necessarily be continuous in such systems. In order to see this, let us take an oscillator with the Hamiltonian
\begin{align} \label{set1}
H=\frac{1}{2}\big(q^2 + p^2\big),
\end{align}
that is, the mass and the frequency are equal to $1$, so that $\beta \hbar$ is dimensionless.

Now, performing a process as a result of which the frequency of the oscillator changes, $\beta^2 [H, U^\dagger H' U]$ will necessarily contain a term proportional to $\beta^2 [p^2, q^2] = -2 i \beta^2 \hbar \{ p, q \}$, where $\{ \cdot, \cdot \}$ is the anticommutator. Introducing the creation and annihilation operators $a$ and $a^\dagger$, we will have $\beta^2 [p^2, q^2] = 2 (\beta\hbar)^2 \left[(a^\dagger)^2 - a^2 \right]$. On the one hand, if $\{ \ket{n} \}_{n=0}^\infty$ are the eigenvectors of $a^\dagger a$, then
\begin{align} \nonumber
\bra{n+2} \beta^2 [p^2, q^2] \ket{n} = 2 (\beta\hbar)^2 \sqrt{(n+1)(n+2)},
\end{align}
meaning that the norm of $\beta^2 [p^2, q^2]$ is infinity for any, no matter how small, non-zero value of $\hbar \beta$.

On the other hand, if $\beta \hbar = 0$, $\Vert \beta^2 [p^2, q^2] \Vert = 0$, which means that the $\beta \hbar \to 0$ limit is not continuous for $\beta^2 [H, U^\dagger H' U]$, which, in turn, implies that the na\"{i}ve $\beta \hbar \to 0$ limit in Eq.~\eqref{eq:bch} has to be taken with extra care. We do this calculation for a specific illustrative example below and show that indeed $\av{e^{-\beta \Omega}}$ converges to $e^{-\beta (F_\beta [H'] - F_\beta [H])}$ as $\beta \hbar \to 0$.

Our goal will thus be to explicitly calculate 
\begin{align} \label{WO1}
\Omega = U^\dagger H' U - H
\end{align}
and then calculate
\begin{align} \label{jarz1}
\big\langle e^{-\beta \Omega} \big\rangle_{\tau_\beta} = \sum_n \bra{\psi_n} \tau_\beta \ket{\psi_n} \, e^{-\beta W_n},
\end{align}
where $W_n$ are the eigenvalues of $\Omega$ and $\ket{\psi_n}$ are its eigenvectors. Then, we are going to take the $\beta \hbar \to 0$ limit. We know that, for any finite $\beta \hbar$, $\av{e^{-\beta \Omega}}_{\tau_\beta}$ will be larger than $e^{-\beta (F_\beta[H'] - F_\beta [H])}$ \cite{Allahverdyan_2005}.


Given the infinite-dimensional Hilbert space, one should of course distinguish between $\Omega$'s with spectra bounded and unbounded from below. In this section, we consider only $\Omega$'s whose spectra are bounded from below.

Since during the driving the Hamiltonian of the oscillator remains quadratic in $q$ and $p$, the generated unitary evolution operator $U$ is Gaussian. Therefore, the transformation of $q$ and $p$ in the Heisenberg picture will be symplectic \cite{Adesso_2007}, which means that, introducing the column vectors $\mathbf{z} = (q, p)$ and $\mathbf{z}' = (U^\dagger q U, \, U^\dagger p U)$, we can write
\begin{align}
\mathbf{z}' = S \mathbf{z},
\end{align}
where $S$ is a $2 \times 2$ symplectic matrix that depends only on $U$.

Now, if
\begin{align}
H' = \mathbf{z}^\rmT \Lambda \mathbf{z},
\end{align}
where $\Lambda > \nul_2$ is a $2 \times 2$ real, symmetric matrix ($\nul_2$ is the $2 \times 2$ zero matrix), then
\begin{align}
U^\dagger H' U = (\mathbf{z}')^\rmT \Lambda \mathbf{z}' = \mathbf{z}^\rmT S^\rmT \Lambda S \mathbf{z},
\end{align}
and therefore, recalling Eq.~\eqref{set1},
\begin{align}
\Omega = \mathbf{z}^\rmT \Xi \mathbf{z}
\end{align}
with
\begin{align} \label{voskor}
\Xi = S^\rmT \Lambda S - \frac{1}{2} \id_2,
\end{align}
where $\id_2$ is the $2 \times 2$ identity matrix. As mentioned above, we will consider the only the processes for which
\begin{align}
\Xi > \nul_2.
\end{align}
In which case, by the Williamson's theorem \cite{Williamson_1936, Adesso_2007}, there exists a symplectic transformation, $\Sigma$, that takes $\mathbf{z}$ to new coordinates $\tilde{\mathbf{z}} = (\tilde{q}, \tilde{p})$, i.e.,
\begin{align} \label{tran1}
\tilde{\mathbf{z}} = \Sigma \mathbf{z},
\end{align}
in which $\Xi$ is proportional to the identity. Namely, if $w > 0$ is the symplectic eigenvalue of $\Xi$, then
\begin{align} \label{Wdiag1}
\Omega = \tilde{\mathbf{z}}^T (\Sigma^{-1})^T \Xi \Sigma^{-1} \tilde{\mathbf{z}} = w (\tilde{q}^2 + \tilde{p}^2).
\end{align}
Introducing the symplectic form $I = \Big(\! \begin{array}{cc} 0 & 1 \\ -1 & 0 \end{array} \!\Big)$, one can calculate $w$ by noting that the eigenvalues of $\Xi I$ are $\pm i w$. Using that, it is easy to show that
\begin{align} \label{magi}
w = \det (\Xi^{1/2}) = \sqrt{\det \Xi}.
\end{align}
Moreover, the symplectic transformation in Eq.~\eqref{Wdiag1} is given by
\begin{align} \label{agic}
\Sigma = w^{-1/2} \Xi^{1/2}.
\end{align}

In order to calculate the trace in Eq.~\eqref{jarz1}, we will make use of some ready formulas for Gaussian states, presented in Ref.~\cite{Scutaru_1998}. The results Ref.~\cite{Scutaru_1998} are given for states, therefore we will introduce density operator
\begin{align}
\kappa_\beta = \frac{1}{Z_W} e^{- \beta \Omega},
\end{align}
where
\begin{align} \label{zeewa}
Z_W = \tr e^{-\beta \Omega} = \left(e^{\beta\hbar w} - e^{-\beta\hbar w}\right)^{-1},
\end{align}
so that
\begin{align} \label{tolave}
\big\langle e^{-\beta \Omega} \big\rangle_{\tau_\beta} = Z_W \tr \big( \kappa_\beta \tau_\beta \big).
\end{align}

Now, since both $\Omega$ and $H$ are purely quadratic, we can write
\begin{align}
\bra{q} \kappa_\beta \ket{\bar{q}} &= e^{-A_W q^2 - A^*_W \bar{q}^2 - 2B_W q \bar{q} + C_W},
\\
\bra{q} \tau_\beta \ket{\bar{q}} &= e^{-A_H q^2 - A^*_H \bar{q}^2 - 2 B_H q \bar{q} + C_H},
\end{align}
and in Ref.~\cite{Scutaru_1998} it is shown that
\begin{align} \label{joint1}
\bra{q} \kappa_\beta \tau_\beta \ket{\bar{q}} = e^{-A q^2 - D \bar{q}^2 - 2B q \bar{q} + C},
\end{align}
with
\beaa \label{joint2}
A &= A_W - \frac{B_W^2}{A_W^* + A_H},
\\
D &= A_H^* - \frac{B_H^2}{A_W^* + A_H},
\\
B &= - \frac{B_W B_H}{A_W^* + A_H},
\\
C &= C_W + C_H + \frac{1}{2} \ln \frac{\pi}{A_W^* + A_H}. ~~
\eeaa

Along with Eq.~\eqref{joint2}, Eq.~\eqref{joint1} immediately yields
\begin{align} \nonumber
\tr (\kappa_\beta \tau_\beta) &= \int\limits_{-\infty}^\infty dq \bra{q} \kappa_\beta \tau_\beta \ket{q}
\\ \label{jarzyn}
&= \frac{\pi e^{C_W + C_H}}{\sqrt{\phantom{\big|}\! |A^*_W + A_H|^2 - (B_W + B_H)^2}}.
\end{align}

All that is left to do now is to determine how $A_W$, $B_W$, and $C_W$ are expressed in terms of $\Xi$. We again use formulas from Ref.~\cite{Scutaru_1998}: If
\begin{align} \nonumber
x^W = \left(\! \begin{array}{cc} x^W_{qq} & x^W_{qp} \\ \!\phantom{\Big|}\! x^W_{qp} & x^W_{pp} \end{array} \!\right),
\end{align}
with the elements defined as
\begin{align} \label{corr1}
x^W_{qq} = 2\av{q^2}_{\kappa_\beta}, \quad\; x^W_{pp} = 2\av{p^2}_{\kappa_\beta}, \quad\; x^W_{qp} = \av{\{ q, p \}}_{\kappa_\beta},
\end{align}
then
\beaa \label{shlor}
A_W &= \frac{1 + \hbar^{-2} \det x^W}{4 \hbar^{-1} x_{qq}^W} - i \frac{x_{qp}^W}{2 x^W_{qq}}, ~~~
\\
B_W &= \frac{1 - \hbar^{-2} \det x^W}{4 \hbar^{-1} x_{qq}^W},
\\
C_W &= - \ln\sqrt{\pi \hbar^{-1} x_{qq}^W}.
\eeaa
Obviously, identical formulas apply for $H$.

Let us now calculate the correlators in Eq.~\eqref{corr1}. In view of Eq.~\eqref{tran1},
\begin{align}
\mathbf{z} = \Sigma^{-1} \tilde{\mathbf{z}}.
\end{align}
So, keeping in mind Eq.~\eqref{Wdiag1}, by introducing creation-annihilation operators, and doing the simple algebra, we arrive at
\begin{align} 
x^W = \hbar \coth(\hbar \beta w) (\Sigma^{-1})^\rmT \Sigma^{-1},
\end{align}
When taking Eqs.~\eqref{magi} and~\eqref{agic} into account, this formula amounts to
\begin{align} \label{corr2}
x^W = \hbar \sqrt{\det \Xi} \coth\Big(\hbar \beta \sqrt{\det \Xi}\Big) \Xi^{-1}.
\end{align}
Analogously, reading the $x$ matrix for $H$ off Eq.~\eqref{corr2}, we get
\begin{align} \label{corr3}
x^H = \hbar \coth\Big(\frac{\beta \hbar}{2}\Big) \id_2.
\end{align}

It is easy to check that, in the $\beta \hbar \to 0$ limit,
\begin{align}
x^W = \frac{\beta^{-1}}{\det \Xi} \bigg(\! \begin{array}{cc} \Xi_{22} & -\Xi_{12} \\ -\Xi_{12} & \Xi_{11} \end{array} \!\bigg) + O(\beta \hbar),
\end{align}
where one should keep in mind that $\Xi$ is a symmetric matrix (because $\Lambda$ is). Similarly,
\begin{align}
x^H = 2 \beta^{-1} \id_2 + O(\beta \hbar).
\end{align}
Using these in Eq.~\eqref{shlor} and substituting the obtained quantities into Eq.~\eqref{jarzyn}, through a somewhat tedious algebra, we find that
\begin{align} \label{capote}
\tr(\kappa_\beta \tau_\beta) = \beta \hbar \sqrt{\frac{\det \Xi}{\det(\Xi + \id_2/2)}} + O[(\beta \hbar)^2].
\end{align}
In order to finally calculate $\big\langle e^{-\beta \Omega} \big\rangle_{\tau_\beta}$, let us first notice that, due to Eq.~\eqref{voskor}, $\det(\Xi + \id_2/2) = \det \Lambda$. Next, from Eq.~\eqref{zeewa}, and keeping in mind Eq.~\eqref{magi}, we see that
\begin{align}
Z_W = \frac{(\beta \hbar)^{-1}}{2 \sqrt{\det \Xi}} + O(1).
\end{align}
Substituting all these into Eq.~\eqref{tolave}, we finally obtain
\begin{align}
\big\langle e^{-\beta \Omega} \big\rangle_{\tau_\beta} = \frac{1}{2 \sqrt{\det \Lambda}} + O(\beta \hbar).
\end{align}

Lastly, calculating $Z_\beta[H']$ and $Z_\beta[H]$ in the same way as we calculated $Z_W$, we find that
\begin{align} \nonumber
e^{-\beta (F_\beta[H'] - F_\beta [H])} = \frac{Z_\beta[H']}{Z_\beta[H]} = \frac{1}{2 \sqrt{\det \Lambda}} + O(\beta \hbar),
\end{align}
which means that
\begin{align}
\big\langle e^{- \beta \Omega} \big\rangle_{\tau_\beta} = e^{-\beta (F_\beta[H'] - F_\beta [H])} + O(\beta \hbar).
\end{align}
This expression shows that the HOW scheme for continuous-variable systems is consistent in the classical limit, at least in the somewhat special case of $\Omega > \nul$. The case of $\Omega < \nul$ remains as an interesting open problem.

\end{document}